\documentclass[onecolumn,american]{IEEEtran}
\usepackage[T1]{fontenc}
\usepackage[utf8]{inputenc}
\usepackage{amsmath}
\usepackage{amsthm}
\usepackage{amssymb}
\usepackage{graphicx}

\makeatletter
\theoremstyle{plain}
\newtheorem{thm}{\protect\theoremname}
\theoremstyle{remark}
\newtheorem{rem}[thm]{\protect\remarkname}
\theoremstyle{definition}
\newtheorem{defn}[thm]{\protect\definitionname}
\theoremstyle{plain}
\newtheorem{prop}[thm]{\protect\propositionname}

\makeatother

\usepackage{babel}
\providecommand{\definitionname}{Definition}
\providecommand{\propositionname}{Proposition}
\providecommand{\remarkname}{Remark}
\providecommand{\theoremname}{Theorem}

\begin{document}
\title{A Non-Probabilistic Game-Theoretic Information Theory Which Subsumes Probabilistic Channel Coding}
\author{\IEEEauthorblockN{Cheuk Ting Li}\\
\IEEEauthorblockA{Department of Information Engineering, The Chinese University of Hong Kong, Hong Kong, China \\
Email: ctli@ie.cuhk.edu.hk}}
\maketitle
\begin{abstract}
Probabilistic settings (e.g., vanishing-error channel coding) and non-probabilistic settings (e.g., zero-error channel coding and adversarial channels) were considered two related but different branches of information theory which do not subsume each other. We propose a unifying non-probabilistic information theory based on game theory and dynamic hedging which subsumes the conventional probabilistic channel coding theorem (vanishing error, with or without feedback) and lossless source coding theorem, as well as zero-error and adversarial settings. Coding is modelled as a deterministic game between an encoder and an adversary, where the encoder may purchase insurance with a payoff that depends on the channel outputs. Our framework is based on a generalization of the works by Ville, Dawid, Shafer and Vovk on the game-theoretic formulation of probabilistic concepts, by relaxing the convex pricing cone to a nonconvex downward closed cone, which is precisely the relaxation needed to model information transmission. Pricing downward closed cone is a versatile tool for non-probabilistic coding results that can subsume their probabilistic counterparts, and provides a canonical form for probabilistic channels, adversarial channels and arbitrarily varying channels.
\end{abstract}


\section{Introduction}

Information theory is mostly built upon probability theory. The central settings in information theory, namely source coding and channel coding, are usually defined using probability \cite{shannon1948mathematical}. While there are non-probabilistic models of channel coding (e.g., zero-error channel coding \cite{shannon1956zero,nair2012nonstochastic,nair2013nonstochastic} and Hamming adversarial channels \cite{hamming1950error}), they are vastly different from their probabilistic counterparts, usually have lower capacities, and cannot subsume the familiar probabilistic (vanishing-error) channel coding theorem.\footnote{The (probabilistic) binary symmetric channel has a capacity strictly greater than its zero-error capacity (which is $0$) and the Hamming (bit-flip) adversarial channel capacity \cite{mceliece1977new}. None of these two non-probabilistic models can recover the probabilistic capacity.}  There are works on combining probabilistic and adversarial settings (e.g., arbitrarily varying channels \cite{blackwell1959capacity,ahlswede1978elimination,lapidoth2002reliable,csiszar2002capacity}; see Section \ref{subsec:prev_avc}), though the combined setting still requires probability. Hence, for communication scenarios where probability is a questionable assumption (e.g., in a network where transmission errors are often caused by non-random actions of nodes), we have to resort to the aforementioned non-probabilistic models with a significant penalty on the capacity. It is natural to raise the questions: 

\smallskip{}

\emph{Is probability essential to the definitions of channel coding and source coding? Or can we find new non-probabilistic versions of these coding settings that generalize the probabilistic versions, and recover the rates given by the channel and source coding theorems?}

\smallskip{}

In this paper, we show that probability is \emph{not} necessary to the definitions of the channel coding and source coding settings, by giving alternative, more general definitions of these settings based on game theory instead. For channel coding, instead of generating the channel output $y_{i}$ randomly conditional on the input $x_{i}$ for time $i=1,\ldots,n$, we assume that $y_{i}$ is determined by an adversary. If the adversary is truly arbitrary, no communication is possible. Hence, we provide an insurance to the encoder. Suppose the encoder loses $\$1$ if the decoder fails to decode correctly. We allow the encoder to purchase insurance that pays upon certain behaviors of the channel. For a simple example where $x_{i},y_{i}\in\mathcal{X}$ (where $\mathcal{X}$ is a finite alphabet), the encoder can purchase an insurance (at a certain price) that pays if $y_{i}\neq x_{i}$, similar to a mail insurance that pays when a mail is lost. The encoder chooses the amount of insurance to purchase at each time, in order to dynamically hedge against the risk of decoding error, so that regardless of $y_{1},\ldots,y_{n}$, the loss is always at most $\$\epsilon$. The insurance policy plays the role of the channel. The capacity of this game-theoretic channel is the maximal communication rate that allows $\epsilon\to0$ as $n\to\infty$. 

Even though this setting is non-probabilistic ($y_{1},\ldots,y_{n}$ are non-random), it generalizes conventional noisy channel coding. Hence, our capacity theorem for the game-theoretic channel (Theorem \ref{thm:cap}) subsumes the conventional channel coding theorem. This is an application of the rather surprising idea that deterministic games can model probabilistic phenomena, which originated from the game-theoretic formulation of the law of large numbers by Ville \cite{ville1939etude}, the prequential framework by Dawid and Vovk \cite{dawid1984present,dawid1999prequential}, and the generalized coin tossing game by Shafer and Vovk \cite{shafer2005probability,shafer2019game}.  Also see \cite{vovk2005defensive,shafer2010base}.

The game-theoretic channel not only generalizes conventional discrete memoryless channels, but also zero-error channels, channels with feedback, and arbitrarily varying channels \cite{blackwell1959capacity,ahlswede1978elimination,lomnitz2013universal} with causal feedback to the adversary. Hence, our theorem which bounds the capacity of game-theoretic channels (Theorem \ref{thm:cap}) applies to all these channels. The insurance policy is a ``canonical form'' of a general class of channels. To accommodate communication settings, we utilize a set of payoff functions of the insurance, called a \emph{pricing downward closed (DC) cone}, which is more general than the convex pricing cones in \cite{shafer2005probability,shafer2010base,shafer2019game}. Pricing DC cones can be nonconvex, which is precisely the relaxation needed to capture information transmission. The pricing DC cone plays three roles simultaneously: probability, adversary and feedback. It is a unifying canonical form that allows these channels to be compared, and equivalent channels to be identified. We also show that lossless source coding can also be modelled as a game using pricing DC cones (Section \ref{sec:lossless}). Hence, two central theorems in information theory can be stated using non-probabilistic games. 

Although this paper is motivated by the theoretical question on whether we can replace probability with game theory in information theory, the channel coding game may be applicable to practical insurance policies. A traditional example is mail insurance which compensates the sender if a mail is lost. For another example, a satellite internet company may want to be insured against the event of poor channel quality (e.g., due to weather).  Companies whose profit depends on connection stability, such as high-frequency trading, may also purchase such insurance. This scenario is elaborated below.

Suppose a high-frequency trader is renting a  connection to an office at the stock exchange. To reduce latency, the connection uses a protocol without automatic retransmission (e.g., UDP), so the trader is responsible for error correction. The network provider, who maintains the connection, wants to assure the trader that the connection is reliable and worth the rent. One method is to have the network provider announce that the packet corruption probability is at most $p$. However, the practical meaning of this statement is unclear, as it is still possible to have a large number of packet corruptions by chance. For the trader to be able to take action against the network provider, the trader would have to record a large number (e.g., $1000$) of packets, and use statistical tests to show that the corruption probability is larger than $p$. Even if the network provider guarantees that there are at most $1000p$ corrupted packets per $1000$ consecutive packets, there is no guarantee on where these corrupted packets will be (e.g., burst error or even adversarial error), so the trader must perform error correction according to the worst case, which leads to a lower transmission rate. The guarantee is also meaningless if the trader wants to send a message that is shorter than $1000$ packets.

A more reliable method is to have the network provider offer an insurance for each packet that pays if that packet is corrupted. If the trader pays a premium $\$\gamma$ for a packet, then if the packet is corrupted, the compensation would be $\$\gamma/p$. By dynamically choosing when to purchase the insurance, the trader can hedge against the loss caused by transmission error, so that the trader can transmit at a rate approaching the capacity of the memoryless packet erasure channel, or otherwise be adequately compensated. The trader can guarantee a profit regardless of the packet corruption pattern. The techniques in this paper reveal how the insurance policy corresponds to the channel model, and how the trader should design the hedging strategy. 

To summarize the key messages of this paper:
\begin{itemize}
\item \emph{Probability can be a consequence rather than an assumption.} Even though the channel coding game is deterministic, it subsumes the probabilistic channel coding theorem (Section \ref{sec:martingale}, Theorem \ref{thm:cap}). Our goal is  to provide an alternative, more general framework that does not rely on the assumption that the channel is probabilistic, but can still recover probabilistic channels if they are desired. 
\item \emph{Pricing DC cone is an abstract language for coding settings.} Pricing DC cones can not only model probabilistic/adversarial channels with/without feedback, but can also model decoding requirements. The calculus of DC cones provides a unifying abstract language that allows us to state entire coding settings as equations. For example, the channel coding setting can be expressed as \emph{one equation} (\ref{eq:channel_require})  that completely captures the dynamics between the encoder, decoder and adversary. 
\item \emph{Causal encoder and adversary are duals of each other.} There is a precise duality between causal actions of the encoder (e.g., channels with feedback) and causal actions of the adversary, described by the dual DC cone. See Sections \ref{subsec:adv_feedback}, \ref{subsec:avcf} and \ref{sec:adv_cost}. 
\end{itemize}
This paper is organized as follows. In Section \ref{sec:previous}, we review some previous works. Section \ref{sec:mail} describes a motivating example using mail insurance. In Section \ref{sec:channelcoding}, we define channel coding as a deterministic game. In Section \ref{sec:martingale}, we explain why the deterministic game generalizes probabilistic channel coding. In Section \ref{sec:dc_cones}, we introduce various operations of pricing DC cones. In Section \ref{sec:capacity}, we state the main result on the capacity of the channel coding game. In Section \ref{sec:cases}, we study various conventional channels that are special cases of the game. In Section \ref{sec:single_cone}, we discuss channels that are characterized by a single cone. Section \ref{sec:adv_cost} describes the dual channel coding setting, where the cost function is chosen by the adversary. In Section \ref{sec:lossless}, we discuss lossless source coding as a game.

\smallskip{}

\begin{rem}
The framework in this paper can be regarded as a ``non-probabilistic information theory'', in the sense that the channel/source coding \emph{settings} are deterministic and do not involve probability. This does not mean that the \emph{proofs} in this paper must not involve probability. Probabilistic constructions are powerful tools even for non-probabilistic problems (e.g., in graph theory), and it is not a meaningful endeavor to avoid them. In this paper, the proofs mostly utilize the calculus of DC cones developed in this paper (which is non-probabilistic), with some steps using probabilistic tools if needed.

In comparison, avoiding probabilistic \emph{settings} can be due to meaningful practical needs (e.g., the channel noise coming from other actors rather than being random, the need of a completely sure guarantee instead of a ``high probability'' one, etc.), which are motivations for the study of zero-error and adversarial coding settings. We also emphasize that the framework in this paper is capable of recovering probabilistic settings, so probability is not completely excluded, but rather becomes an optional interpretation. 
\end{rem}
\smallskip{}

\subsection*{Notations}

 Entropy is in bits, and logarithms are to the base $2$. For a statement $S$, its indicator is $\mathbf{1}\{S\}$, which is $1$ if $S$ holds, or $0$ if $S$ does not hold. Write $\mathbb{R}_{\ge0}:=\{x\in\mathbb{R}:x\ge0\}$, and similar for $\mathbb{R}_{>0}$, $\mathbb{Z}_{\ge0}$, etc. For a sequence $x_{1},\ldots,x_{n}$, write $x^{n}=(x_{1},\ldots,x_{n})$, and $x_{a}^{b}=(x_{a},\ldots,x_{b})$.  The set of functions from $\mathcal{X}$ to $\mathcal{Y}$ is denoted as $\mathcal{Y}^{\mathcal{X}}$. Given $A,B\subseteq\mathbb{R}^{\mathcal{X}}$, their Minkowski sum is $A+B:=\{a+b:\,a\in A,b\in B\}$ (where $a+b$ is the function $x\mapsto a(x)+b(x)$), and we write $\gamma A:=\{\gamma a:\,a\in A\}$ for $\gamma\in\mathbb{R}$, and $\mathrm{hull}(A)$ is the convex cone (set of finite convex combinations) of $A$. The  set of finite discrete probability mass functions over $\mathcal{X}$ is $\Delta_{\mathcal{X}}:=\{p\in\mathbb{R}_{\ge0}^{\mathcal{X}}:|\mathrm{supp}(p)|<\infty,\,\sum_{x\in\mathrm{supp}(p)}p(x)=1\}$, where $\mathrm{supp}(p):=\{x\in\mathcal{X}:\,p(x)\neq0\}$. For $x\in\mathcal{X}$, let $\mathbf{e}_{x}\in\Delta_{\mathcal{X}}$ be the vector with entries $(\mathbf{e}_{x})_{t}=\mathbf{1}\{t=x\}$ for $t\in\mathcal{X}$. For $a,b\in\mathbb{R}^{\mathcal{X}}$, write $\langle a,b\rangle:=\sum_{x\in\mathrm{supp}(a)\cap\mathrm{supp}(b)}a(x)b(x)$. \smallskip{}

\section{Previous Works}\label{sec:previous}

\subsection{Generalized Coin Tossing Game}\label{subsec:prev_coin}

This paper follows a similar spirit as \cite{shafer2005probability,shafer2019game}, which models various probabilistic and financial concepts as non-probabilistic games. In \cite{shafer2005probability}, a stochastic process is likened to a generalized coin tossing game between three players: forecaster, skeptic and reality. The reality (adversary) chooses the outcomes of the process. The purpose of the skeptic's action (which lies in a convex cone called a \emph{pricing cone} \cite{shafer2010base}) is to ensure that either the outcomes are favorable, or a large gain can be attained. 

In the channel coding game in this paper, the encoder plays the role of the skeptic whose goal is to ensure successful decoding or a large gain. The important difference is that the encoder's action lies in the pricing DC cone, which may not be convex. Eliminating the convexity assumption is the key step that changes the game from a ``passive'' process in \cite{shafer2005probability} (even though the skeptic is an active player, the skeptic's role is merely to stop reality from choosing unlikely outcomes) to an ``active'' decision process which can be used to transmit information. The encoder is not only a skeptic, but can influence the process. The ``game-theoretic information theory'' in this paper is not merely an application of the game-theoretic probability in \cite{shafer2005probability,shafer2019game}, but can be seen as an extension in the sense of having a more active skeptic/encoder.

\smallskip{}

\subsection{Adversarial Channels and Arbitrarily Varying Channels}\label{subsec:prev_avc}

In an adversarial channel \cite{hamming1950error,shannon1956zero}, the noise sequence is not random, but is rather controlled by an adversary. If the adversary can choose each entry of the noise sequence separately, this is equivalent to zero-error channel coding \cite{shannon1956zero}. Another example is the Hamming bit-flip adversarial channel where the channel input is a bit sequence, and the adversary can flip at most a certain portion of the bits, which is equivalent to finding a binary code with a lower bound on the minimum distance \cite{hamming1950error,gilbert1952comparison,varshamov1957estimate,mceliece1977new}. A variant of the Hamming adversarial channel, where the adversary observes the input sequence causally, has also been studied in \cite{langberg2009binary,chen2015characterization}. 

Adversarial channels have vastly different behaviors compared to probabilistic channels. Coding against an adversary is significantly more difficult compared to random noise, and the capacity of adversarial channels are often difficult to compute and strictly smaller than their probabilistic counterparts \cite{mceliece1977new,lovasz1979shannon}. There is usually no direct correspondence between results on non-probabilistic adversarial channels (or zero-error channel coding) and results on probabilistic channels, other than the loose bound that adversarial channels have lower capacities.

Arbitrarily varying channels (AVCs) \cite{blackwell1959capacity,ahlswede1978elimination,lapidoth2002reliable,csiszar2002capacity} are channels where the output sequence depends stochastically on the input sequence and a state sequence, and the state sequence can vary in an arbitrary manner. An AVC can be regarded as a game between the encoder-decoder team and an adversary who chooses the state sequence. Refer to \cite{dey2024codes} for various generalizations of AVCs.

Another related setting is the mutual information game \cite{mceliece1983communication,borden1985some}, where the encoder controls the input distribution $p_{X}$ and seeks to maximize $I(X;Y)$, whereas the adversary partially controls the channel $p_{Y|X}$ and seeks to minimize $I(X;Y)$. This game has found applications in jamming channels \cite{medard1997capacity,stark2002capacity,shafiee2009mutual}. Also see \cite{jose2018game,jose2020shannon,vora2019minimax} for related works on game-theoretic settings with adversarial jammers, and \cite{harremoes2002unified} for a game-theoretic treatment for several problems in Shannon theory. 

Our setting is different from these works in three aspects. First, AVCs and jamming channels are stochastic, whereas the game in this paper is deterministic. Our focus is recovering probabilistic channels using deterministic games, which has not been studied in these previous works. Second, the payoff in our setting is the monetary gain from the insurance (can be positive or negative), which has a different operational meaning compared to the payoffs in previous works which are either the channel capacity or the error probability (must be nonnegative).\footnote{The payoff in \cite{harremoes2002unified} can be entropy, relative entropy or other information quantities.} Third, in our setting, the adversary can choose the channel outputs (not only the state) arbitrarily. The encoder cannot stop the adversary from choosing irrelevant outputs, but can only be compensated by the insurance if this happens. We also show that our game generalizes AVCs with causal feedback to the adversary (Section \ref{subsec:avcf}).

Readers are referred to \cite{akyol2016information,khouzani2018optimal,hernandez2014nash} for other works on the overlap between information theory and game theory.

\smallskip{}

\subsection{Nonstochastic Information Theory}

A nonstochastic information theory has been studied in \cite{nair2012nonstochastic,nair2013nonstochastic} (also see the generalization in \cite{rangi2019non}), which concerns variables with nonstochastic uncertain values rather than random values, and recovers the zero-error channel capacity (with or without feedback) via the maximin information functional. In comparison, the framework in this paper recovers both the zero-error and the vanishing-error channel capacity (with or without feedback).  A key difference is that we allow purchasing insurance on the values of the variables, which provides finer control and lets us recover probabilistic concepts. Since our approach is fundamentally different from \cite{nair2012nonstochastic,nair2013nonstochastic}, we will not call our approach ``nonstochastic information theory'' to avoid confusion.

Information can be modelled as a partition of the sample space \cite{ellerman2014introduction}, or by a confusion hypergraph \cite{korner1971coding,li2025coding}. While these notions are not inherently probabilistic, and their operations can be defined without a probability distribution, probability is still required to recover the conventional coding theorems.

\section{Motivating Example: Mail Insurance}\label{sec:mail}

We use a simple, familiar scenario of mail insurance to motivate the channel coding game. Suppose the sender can send $n$ mails, each consisting of $\ell$ bits. The sender wants to convey a message with $k\ell$ bits ($1\le k\le n$) to the receiver. To this end, the sender encodes the message into the $n$ mails so that any $k$ of the mails can recover the message (this is possible using a linear code). If the receiver fails to correctly decode the message, the sender-receiver team loses $\$1$. If each mail is lost with a certain positive probability, then the worst-case loss is $\$1$ regardless of how small the probability is. Nevertheless, if an insurance is available, then we can guarantee a better worst-case loss without the need of probability.

For each mail being sent, the post office offers an insurance at a price $\$1$, which compensates $\$(1/p)$ if the mail is lost (where $0<p<1$). The sender can purchase any nonnegative amount $\gamma$ of insurance. If the sender pays a premium $\$\gamma$, then the compensation would be $\$\gamma/p$ if the mail is lost. The sender's goal is to use the insurance to hedge against the risk of decoding error, so they will lose at most $\$\epsilon$ regardless of which mails are lost, where $0<\epsilon<1$ is as small as possible.

One simple strategy is to pay a premium $\$p/(n-k+1)$ for each mail. If no decoding error occurs, the sender loses at most $\$np/(n-k+1)$. If decoding error occurs (which costs $\$1$), then there are at least $n-k+1$ lost mails, so the net gain is at least
\[
\frac{p}{n-k+1}\Big(\frac{n-k+1}{p}-n\Big)-1=-\frac{np}{n-k+1}.
\]
Hence, the maximum loss is $\$np/(n-k+1)$. If $k=\lfloor nR\rfloor$ where the rate $R$ satisfies $R<1-p$, the maximum loss is less than $\$1$, so the insurance can reduce the maximum loss. 

A better strategy is to perform dynamic hedging \cite{ville1939etude,shafer2005probability}, where the sender sends the $n$ mails one by one, and adaptively chooses the amount of insurance to purchase based on the number of mails successfully delivered. If $k$ mails have already been successfully delivered, then there is no need to purchase more insurance. If only $k-1$ mails have been successfully delivered, and there are only a few mails left to be sent, then a larger amount of insurance should be purchased. 

One can show that, taking $k=\lfloor nR\rfloor$ and $n\to\infty$, as long as $R<1-p$, the maximum loss $\epsilon$ can approach $0$. This is reminiscent of the channel coding theorem for the erasure channel with erasure probability $p$ (which is the implied probability of the insurance policy which compensates $\$1/p$ if the mail is lost), where the capacity is also $1-p$. We will show that there is indeed a rigorous correspondence between probabilistic channels and certain deterministic games, to be introduced in the following sections.

\section{The Channel Coding Game}\label{sec:channelcoding}

A channel is a mechanism mapping an input $x\in\mathcal{X}$ to an output $y\in\mathcal{Y}$. In conventional information theory, $y$ is chosen by nature randomly conditional on $x$. Since our goal is to develop a non-probabilistic information theory, nature is considered an adversary which chooses $y$ arbitrarily. If $y$ is truly arbitrarily and cannot be influenced by the encoder, then no information can be transmitted. Suppose the encoder loses $\$1$ if a decoding error occurs. Without any insurance, the encoder will surely lose $\$1$. Hence, we allow the encoder to purchase insurance to hedge against the risk of undesirable $y$'s. 

Consider the simple example where $\mathcal{Y}=\{0,1\}$. Fix $0<\beta<1/2$. At time $i=1,\ldots,n$, the encoder may purchase one of the two insurance policies: Policy 0 ``if $y_{i}\neq0$, then the encoder gains $\$1/\beta$'', or Policy 1 ``if $y_{i}\neq1$, then the encoder gains $\$1/\beta$'', for $\$1$ each. The encoder may purchase an arbitrary nonnegative amount of Policy 0, or an arbitrary nonnegative amount of Policy 1, but not both. We assume the existence of feedback, so the encoder's choice can depend on $y^{i-1}=(y_{1},\ldots,y_{i-1})$. After the encoder has chosen the amount of insurance to purchase at time $i$, the adversary chooses $y_{i}$ accordingly. The encoder's goal is to lose at most $\$\epsilon$ (where $0<\epsilon<1$ is small), regardless of $y_{1},\ldots,y_{n}$. This means the encoder can price the transmission service at any amount greater than $\$\epsilon$, and have a guaranteed profit.

Loosely speaking, the strategy of the encoder is to use the insurance to force the adversary to choose a desirable $y^{n}$. Suppose the decoder can decode correctly as long as $y^{n}$ is close to a certain sequence $\tilde{y}^{n}$. At time $i$, the encoder purchases Policy 0 if $\tilde{y}_{i}=0$, or Policy 1 if $\tilde{y}_{i}=1$. If $y^{n}$ is close to $\tilde{y}^{n}$, then the encoder will have a small loss in the insurance, but there is no decoding error so the overall loss is small. If $y^{n}$ is far from $\tilde{y}^{n}$, then the encoder will gain from the insurance, which offsets the loss incurred by decoding error so the overall loss is also small.

Generally, the set of possible payoffs of the insurance policy is described by the following concept, which is a relaxation of the notion of convex pricing cones in \cite{shafer2005probability,shafer2010base}.

\smallskip{}

\begin{defn}
\label{def:dc_cone}Given a set $\mathcal{Y}$, a \emph{pricing downward closed (DC) cone} over $\mathcal{Y}$ is a   set $A\subseteq\mathbb{R}^{\mathcal{Y}}$ that is a cone (i.e., $\gamma A\subseteq A$ for every $\gamma\ge0$), and is downward closed (i.e., $A+\mathbb{R}_{\le0}^{\mathcal{Y}}\subseteq A$). In other words, if $a\in A$, $b\in\mathbb{R}_{\ge0}^{\mathcal{Y}}$ and $\gamma\ge0$, then $\gamma a-b\in A$. Note that $\emptyset$ is a pricing DC cone. The set of all pricing DC cones over $\mathcal{Y}$ is denoted as $\mathrm{DCCone}(\mathcal{Y})$.
\end{defn}
\smallskip{}

The elements $a\in A$ of a pricing DC cone are possible payoff functions, which we call \emph{portfolios}. The encoder is allowed to choose any portfolio $a\in A$. Then, if the channel output is $y$, the encoder gains $\$a(y)$. $A$ is a cone, meaning that the encoder can always scale the portfolio by $\gamma\ge0$. $A$ being downward closed means that if the encoder can choose $a\in A$, then the encoder can also choose a portfolio that always pays less than $a$.\footnote{This is called \emph{slackening} for convex pricing cones \cite{shafer2019game}.} An important difference between pricing DC cones and convex pricing cones (which are also downward closed) \cite{shafer2005probability,shafer2010base} is that a pricing DC cone $A$ may not be convex, and $a,b\in A$ does not imply $a+b\in A$, i.e., we do not always allow combining portfolios. For example, the pricing DC cone for the $\mathcal{Y}=\{0,1\}$ example earlier in this section is
\begin{align}
A=\{a\in\mathbb{R}^{\{0,1\}}: & (1-\beta)a(0)+\beta a(1)\le0\nonumber \\
 & \mathrm{or}\;\;\beta a(0)+(1-\beta)a(1)\le0\},\label{eq:bsc_feedback}
\end{align}
which is nonconvex. We have $(-1,1/\beta-1),(1/\beta-1,-1)\in A$, meaning that the encoder can purchase the policy ``if $y_{i}\neq0$, then the encoder gains $\$1/\beta$'' or the policy ``if $y_{i}\neq1$, then the encoder gains $\$1/\beta$'' for $\$1$ each, but not both. See Remark \ref{rem:combination} for the reason why having nonconvex pricing DC cones and disallowing combination of portfolios are reasonable, and see Figure \ref{fig:cones} for an illustration.

More generally, we allow the insurance policy to depend on the channel input $x$. It is described by a function $W:\mathcal{X}\to\mathrm{DCCone}(\mathcal{Y})$, meaning that if the channel input is $x$, then we allow the encoder to choose a portfolio in $W(x)$.  We now define channel coding as a game, inspired by the generalized coin tossing game \cite{shafer2005probability} (with the important difference being that the pricing DC cone can be nonconvex). We will see that the role of the channel is played by the insurance policy $W$. For example, the policy (\ref{eq:bsc_feedback}) is the game-theoretic version of the binary symmetric channel $\mathrm{BSC}(\beta)$ with feedback.

\smallskip{}

\begin{defn}
\label{def:game_channel} A \emph{game-theoretic channel} from a set $\mathcal{X}$ to a set $\mathcal{Y}$ is a function $W:\mathcal{X}\to\mathrm{DCCone}(\mathcal{Y})$.    The \emph{channel coding game} is parametrized by a tuple $(W,n,L,\epsilon)$, where $W$ is a game-theoretic channel, $n\in\mathbb{Z}_{>0}$ is the blocklength, $L\in\mathbb{Z}_{>0}$ is the message cardinality, and $0<\epsilon<1$ is the maximum loss. It is a game with three players (adversary, encoder and decoder; where the encoder and decoder are a team, but cannot directly communicate) defined as follows:
\begin{itemize}
\item Adversary produces $m\in[L]$, which is observed by  the encoder.
\item Encoder produces $x_{1},\ldots,x_{n}\in\mathcal{X}$, which is observed by the adversary.
\item For $i=1,\ldots,n$:
\begin{itemize}
\item Encoder produces $w_{i}\in W(x_{i})$, which is observed by the adversary.
\item Adversary  produces $y_{i}\in\mathcal{Y}$, which is observed by the encoder.
\end{itemize}
\item Decoder observes $y_{1},\ldots,y_{n}$, produces $\hat{m}\in[L]$.
\item The team wins\footnote{We may also require that $\sum_{j=1}^{i}w_{j}(y_{j})\ge-\epsilon$ for all $i\in[n]$ (or else the team loses), i.e., the team can never be in a deficit larger than $\epsilon$ at any time (similar to \cite{shafer2005probability}). The main results (Theorems \ref{thm:cap}, \ref{thm:avcf}) continue to hold in this case.} if 
\[
\sum_{i=1}^{n}w_{i}(y_{i})-\mathbf{1}\{\hat{m}\neq m\}\ge-\epsilon,
\]
that is, the gains of the insurance, subtracted by the $\$1$ loss in case of decoding error, is at least $-\epsilon$.
\end{itemize}
\end{defn}
\smallskip{}

\begin{figure}
\centering
\includegraphics[scale=0.95]{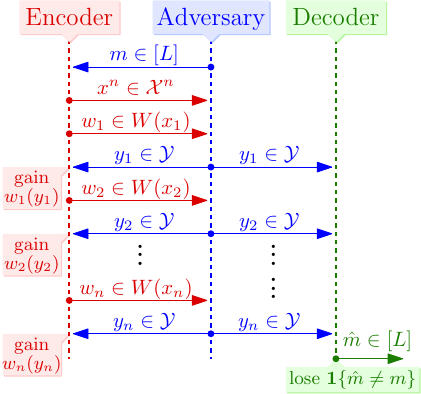}

\caption{Sequence diagram of the channel coding game (Definition \ref{def:game_channel}).}\label{fig:game}
\end{figure}

A \emph{strategy} of a player at a step in the game is a function that maps all previous observations  of the player to the action of the player at that step, and a strategy of a team is a tuple of strategies of each player in the team at every step. Here, a strategy of the  team consists of a function $m\mapsto x^{n}$ used by the encoder, a function $(m,i,y^{i-1})\mapsto w_{i}$ used by the encoder, and a function $y^{n}\mapsto\hat{m}$ used by the decoder.  We say that the game is a \emph{guaranteed win} if there exists a strategy of the team which guarantees a win regardless of the choices of $m,y^{n}$ by the adversary.\footnote{More explicitly, the game is a guaranteed win if there exist functions $x^{n}:[L]\to\mathcal{X}^{n}$, $w_{i}:[L]\times\mathcal{Y}^{i-1}\to\mathbb{R}^{\mathcal{Y}}$ for $i\in[n]$, and $\hat{m}:\mathcal{Y}^{n}\to[L]$, such that $w_{i}(m,y^{i-1})\in W(x_{i}(m))$ for all $i,m,y^{i-1}$ (where $x_{i}(m)$ is the $i$-th entry of $x^{n}(m)$), and $\sum_{i=1}^{n}w_{i}(m,y^{i-1})(y_{i})-\mathbf{1}\{\hat{m}(y^{n})\neq m\}\ge-\epsilon$ for all $m,y^{n}$.} This strategy is called a \emph{winning strategy}. These are deterministic concepts. No probabilities are involved. 

The game is an iterative process where the encoder and the adversary produce $w_{i}$ and $y_{i}$ interactively (see Figure \ref{fig:game}). This might appear unrealistic as the encoder has to make decision quickly at each time step. In Section \ref{sec:dc_cones}, we discuss an equivalent formulation where the encoder makes a decision in one step without the need of feedback.

The interactive nature of this game is reminiscent of channels with feedback, where the encoder can produce $x_{i}$ depending on $y^{i-1}$. However, the channel coding game does not have such a feedback, since the encoder must produce $x_{1},\ldots,x_{n}$ before observing any $y_{i}$. Nevertheless, the channel coding game does include channels with feedback as special cases using non-convex pricing DC cones like (\ref{eq:bsc_feedback}). See Section \ref{sec:martingale}.

We can now define the capacity of a game-theoretic channel.

\smallskip{}

\begin{defn}
[Game-theoretic capacity]\label{def:capacity} Given a game-theoretic channel $W:\mathcal{X}\to\mathrm{DCCone}(\mathcal{Y})$, its \emph{capacity} $\mathrm{C}(W)$ is  the supremum of the set of rates $R\ge0$ such that for every $n_{0}\in\mathbb{Z}_{>0}$, $0<\epsilon<1$, there exists $n\ge n_{0}$ such that the channel coding game $(W,n,\lfloor2^{nR}\rfloor,\epsilon)$ is a guaranteed win.\footnote{If $W(x)=\emptyset$ for all $x$, so no strategy of the team exists, take $\mathrm{C}(W)=-\infty$.}
\end{defn}
\smallskip{}

The main result on the capacity will be given in Section \ref{sec:capacity}. Even though a game-theoretic channel  is only characterized by a function $W:\mathcal{X}\to\mathrm{DCCone}(\mathcal{Y})$, it is surprising general. In Section \ref{sec:cases}, we will show that it includes the following conventional notions of channels as special cases:
\begin{itemize}
\item Discrete memoryless channels without feedback (Section \ref{subsec:dmc}).
\item Discrete memoryless channels with noiseless feedback (Section \ref{subsec:dmc_feedback}).
\item Zero-error channels, with or without feedback (Sections \ref{subsec:adv}, \ref{subsec:adv_feedback}).
\item Arbitrarily varying channels with causal feedback to the adversary (Section \ref{subsec:avcf}). 
\end{itemize}
Hence, any example of the above channels, or combinations thereof, can be succinctly represented by a function $W:\mathcal{X}\to\mathrm{DCCone}(\mathcal{Y})$, which is a canonical representation that allows us to compare channels and identify equivalent channels. While it might seem counter-intuitive that a non-probabilistic game can generalize a probabilistic noisy channel, the precise meaning of this claim will be explained in the next section.

\smallskip{}

\begin{rem}
\label{rem:combination}Pricing DC cones may not be convex. For example, for the $A$ in (\ref{eq:bsc_feedback}), we have $(-1,1/\beta-1),(1/\beta-1,-1)\in A$, but their sum $(1/\beta-2,1/\beta-2)\notin A$. If we allow choosing $(1/\beta-2,1/\beta-2)$ which has positive entries, the encoder can perform arbitrage, i.e., gain $1/\beta-2$ regardless of $y$. Hence, this $A$ might appear unreasonable, as the only reason the encoder cannot perform arbitrage is that we ``artificially'' disallow combining portfolios. 

We now explain why disallowing combination is natural in information transmission. $A$ is the union of two parts $\{a:(1-\beta)a(0)+\beta a(1)\le0\}$ and $\{a:\beta a(0)+(1-\beta)a(1)\le0\}$. The first part corresponds to sending $x=0$ in a $\mathrm{BSC}(\beta)$, whereas the second part corresponds to sending $x=1$. Therefore, combining two portfolios in different parts is like sending $x=0$ and $x=1$ simultaneously, which is not allowed. For a more concrete example, consider the mail insurance in Section \ref{sec:mail}. Assume that the mail insurance not only pays upon lost of mails, but also when the wrong mail is delivered to the recepient. The insured event is ``if the mail is lost, or if the content of the received mail is not this particular string''. The insurer would clearly make sure that the mail being sent matches the string recorded in the insurance, so it is impossible to combine two insurance policies with two different strings. Generally, we disallow combination because choosing a portfolio is analogous to choosing a channel input, and an encoder cannot send two different input symbols at the same time. 
\end{rem}
\smallskip{}

\section{Recovering Probabilistic Channels by Deterministic Games}\label{sec:martingale}

It is unsurprising that the adversarial game in Definition \ref{def:game_channel} generalizes adversarial channels and zero-error channel coding (see Sections \ref{subsec:adv}, \ref{subsec:adv_feedback}). Nevertheless, it is unclear how a deterministic game can generalize a probabilistic channel. In this section, we will use the pricing DC cone in (\ref{eq:bsc_feedback}) as an example to demonstrate how $\mathrm{BSC}(\beta)$ with feedback can be recovered via a generalization of the strategy based on martingale by Ville \cite{ville1939etude} and Shafer and Vovk \cite{shafer2005probability}. We call our generalization the \emph{nonconvex martingale strategy}.

We first review the notion of feedback coding for $\mathrm{BSC}(\beta)$. A feedback coding scheme is given by the encoding functions $f_{n,i}:[L]\times\{0,1\}^{i-1}\to\{0,1\}$ for $i\in[n]$ and the decoding function $g_{n}:\{0,1\}^{n}\to[L]$. Let $M\in[L]$, $\tilde{X}_{i}=f_{n,i}(M,Y^{i-1})$, and $Y_{i}\in\{0,1\}$ is generated conditional on $\tilde{X}_{i}$ according to $\mathrm{BSC}(\beta)$ for $i=1,\ldots,n$ (see Remark \ref{rem:bsc_x_tilde} for the reason we use $\tilde{X}_{i}$ instead of $X_{i}$). We say that the scheme has a maximum error probability $\epsilon$ if $\max_{m\in[L]}\mathbb{P}(g_{n}(Y^{n})\neq m|M=m)\le\epsilon$. 

We now recover $\mathrm{BSC}(\beta)$ with feedback as a channel coding game. Consider $\mathcal{X}=\{0\}$, and $W(0)=A$ where $A$ is given in (\ref{eq:bsc_feedback}). This game-theoretic channel has no explicit inputs $x^{n}$, so the encoder can only influence $y_{i}$ through the choice of the portfolio $w_{i}\in A$. We will prove the following claim, showing that the maximum error probability $\epsilon$ corresponds to the maximum loss $\epsilon$ in the game. The concept of ``error probability'' can be recovered in the deterministic game without probability.

\smallskip{}

\begin{prop}
\label{prop:bsc}There exists a feedback coding scheme for $\mathrm{BSC}(\beta)$ with maximum error probability $\epsilon$ if and only if there exists a winning strategy for the channel coding game $(W,n,L,\epsilon)$ with $W:\{0\}\to\mathrm{DCCone}(\{0,1\})$, $W(0)=A$ where $A$ is given in (\ref{eq:bsc_feedback}).
\end{prop}
\smallskip{}

\begin{IEEEproof}
The game in Definition \ref{def:game_channel} already works in a similar iterative manner as feedback coding. There are two main differences: 1) the game has a non-probabilistic winning guarantee for all $y^{n}$, whereas the decoding requirement in feedback coding is only probabilistic; and 2) the encoder chooses the portfolio $w_{i}\in A$ in the game, whereas the encoder chooses $\tilde{X}_{i}$ in feedback coding. For the ``if'' direction where we construct a feedback coding scheme using a winning strategy of the game, the first difference is not a hindrance since a non-probabilistic guarantee implies a probabilistic guarantee, so we can still assume $Y^{n}$ is random. For the second difference, to see how $w_{i}$ corresponds to $\tilde{X}_{i}$, note that there are two cases in the DC cone $A$ in (\ref{eq:bsc_feedback}): $(1-\beta)a(0)+\beta a(1)\le0$, and $\beta a(0)+(1-\beta)a(1)\le0$. The first case corresponds to $\tilde{X}_{i}=0$, since $\mathbb{E}[a(Y_{i})|\tilde{X}_{i}=0]\le0$ if $(1-\beta)a(0)+\beta a(1)\le0$. The second case corresponds to $\tilde{X}_{i}=1$. Therefore, if the adversary chooses $Y_{i}$ randomly in the game according to $\mathrm{BSC}(\beta)$ with input $\tilde{X}_{i}$, it can guarantee a nonpositive expected payoff in the insurance. The only way to guarantee a win in the game is to ensure that the message is decoded correctly most of the time. 

For a rigorous proof of the ``if'' direction, assume that there exists a winning strategy for the channel coding game. To find the corresponding coding scheme, let $w_{i}(m,y^{i-1})$ be the portfolio produced by the encoder upon the message $m$ and feedback $y^{i-1}$. We have $w_{i}(m,y^{i-1})\in A$, which means that 
\begin{align}
 & (1-\beta)w_{i}(m,y^{i-1})(0)+\beta w_{i}(m,y^{i-1})(1)\le0,\label{eq:ex_bsc_0}\\
 & \mathrm{or}\;\;\beta w_{i}(m,y^{i-1})(0)+(1-\beta)w_{i}(m,y^{i-1})(1)\le0,\label{eq:ex_bsc_1}
\end{align}
or both.  Take $f_{n,i}(m,y^{i-1})=0$ if the first case holds, or $f_{n,i}(m,y^{i-1})=1$ otherwise. The decoding function $g_{n}(y^{n})$ is simply the output of the decoder upon $y^{n}$. To bound the error probability of this scheme, let $M=m\in[L]$, $\tilde{X}_{i}=f_{n,i}(M,Y^{i-1})$, and $Y_{i}\in\{0,1\}$ is generated conditional on $\tilde{X}_{i}$ according to $\mathrm{BSC}(\beta)$. If $\tilde{X}_{i}=0$, we have $\mathbb{E}[w_{i}(m,Y^{i-1})(Y_{i})|\tilde{X}_{i}=0]\le0$ due to (\ref{eq:ex_bsc_0}) and $Y_{i}|\{\tilde{X}_{i}=0\}\sim\mathrm{Bern}(\beta)$. Similarly, we also have $\mathbb{E}[w_{i}(m,Y^{i-1})(Y_{i})|\tilde{X}_{i}=1]\le0$ due to (\ref{eq:ex_bsc_1}). Hence, $\mathbb{E}[w_{i}(m,Y^{i-1})(Y_{i})]\le0$. The game is a guaranteed win, which implies $\mathbb{E}[\sum_{i}w_{i}(m,Y^{i-1})(Y_{i})-\mathbf{1}\{g_{n}(Y^{n})\neq m\}]\ge-\epsilon$. Combining this with $\mathbb{E}[w_{i}(m,Y^{i-1})(Y_{i})]\le0$ gives $\mathbb{P}(g_{n}(Y^{n})\neq m)\le\epsilon$. We remark that the use of probability is restricted to the proof (if the team loses no more than $\$\epsilon$ for every $y^{n}$, then the expected loss is no more than $\$\epsilon$ for random $y^{n}$ as well). We do not assume that the adversary can pick $y_{i}$ randomly in Definition \ref{def:game_channel}.

For the ``only if'' direction, consider a coding scheme $f_{n,i},g_{n}$ for the BSC with maximum error probability $\epsilon$. Define $M,\tilde{X}_{i},Y_{i}$ according to the scheme. Let
\[
P_{e}(m,y^{i})=\mathbb{P}(g_{n}(Y^{n})\neq m\,|\,M=m,\,Y^{i}=y^{i})
\]
be the conditional error probability given $m$ and $y^{i}$. We now design a strategy for the game, which we call the nonconvex martingale strategy. At time $i$, observing the message $m$ and the feedback $y^{i-1}$, the encoder produces $w_{i}$ with 
\[
w_{i}(y')=P_{e}(m,(y^{i-1},y'))-P_{e}(m,y^{i-1}),
\]
where $(y^{i-1},y')\in\{0,1\}^{i}$ denotes the concatenation. See Remark \ref{rem:martingale} for the connection with \cite{ville1939etude,shafer2005probability,shafer2010base}. Intuitively, if $y_{i}=y'$ increases the error probability, the encoder assigns a larger payoff to that value of $y'$ to hedge against that value. To check $w_{i}\in A$, if $f_{n,i}(m,y^{i-1})=0$ (i.e., $\tilde{X}_{i}=0$), then $Y_{i}\sim\mathrm{Bern}(\beta)$, so 
\[
P_{e}(m,y^{i-1})=(1-\beta)P_{e}(m,(y^{i-1},0))+\beta P_{e}(m,(y^{i-1},1))
\]
by the law of total probability, giving $(1-\beta)w_{i}(0)+\beta w_{i}(1)=0$. The case $f_{n,i}(m,y^{i-1})=1$ corresponds to the other case in $A$. The decoder outputs $\hat{m}=g_{n}(y^{n})$. The final payoff is
\begin{align*}
 & \sum_{i=1}^{n}w_{i}(y_{i})-\mathbf{1}\{\hat{m}\neq m\}\\
 & =P_{e}(m,y^{n})-P_{e}(m,\emptyset)-\mathbf{1}\{\hat{m}\neq m\}\,\ge\,-\epsilon,
\end{align*}
since $P_{e}(m,y^{n})=\mathbf{1}\{\hat{m}\neq m\}$ (if we know $m,y^{n}$, then we know whether there is an error), and $P_{e}(m,\emptyset)$ is the conditional error probability given $m$, which is at most $\epsilon$. Hence, the team wins.
\end{IEEEproof}
\smallskip{}

\begin{rem}
\label{rem:martingale}Note that $P_{e}(m,Y^{i})$ forms a martingale, and $w_{i}(Y_{i})$ forms a martingale difference sequence. The strategy we discussed earlier is similar to the strategy based on martingales in \cite{ville1939etude,shafer2005probability,shafer2010base}. The difference is that, since our goal is to use games to model channel coding, we allow the conditional distribution of $Y_{i}$ given $(M,Y^{i-1})$ to be partially controlled by the encoder. Therefore, the DC cone $A$ has to include all valid choices of portfolios (or martingale difference sequences) by the encoder, which causes $A$ to be nonconvex. In comparison, the skeptic in \cite{shafer2005probability,shafer2010base} has a convex action space, and cannot influence $y_{i}$ in the same informative manner.
\end{rem}
\smallskip{}

\begin{rem}
\label{rem:bsc_x_tilde}We use $\tilde{X}_{i}$ for the input to the BSC with feedback instead of $X_{i}$ because the input $\tilde{X}_{i}$ to the BSC does not correspond to the $x_{i}$ in the channel coding game. As we will see later, the inputs $x^{n}$ in the game correspond to the ``noncausal inputs'' that are chosen at the beginning (e.g., the inputs to a BSC without feedback). The ``causal inputs'' which can depend on the feedback, like the $\tilde{X}_{i}$ in this section, are included in the portfolios $w_{i}$.
\end{rem}
\smallskip{}

\begin{rem}
We can also design a game where the maximum loss $\epsilon$ corresponds to the average error probability $L^{-1}\sum_{m=1}^{L}\mathbb{P}(g_{n}(Y^{n})\neq m|M=m)$ instead of the maximum error probability. We insert a step at the beginning of Definition \ref{def:game_channel} where the encoder produces a portfolio $b\in W_{M}$ in the pricing DC cone $W_{M}=\{b\in\mathbb{R}^{[L]}:\,\sum_{m=1}^{L}b(m)\le0\}$, and lets the adversary observe $b$. At the end, the team wins if $b(m)+\sum_{i=1}^{n}w_{i}(y_{i})-\mathbf{1}\{\hat{m}\neq m\}\ge-\epsilon$.
\end{rem}
\smallskip{}

\section{The Calculus of Pricing DC Cones}\label{sec:dc_cones}

\subsection{Examples of DC Cones}

We will introduce the machinery needed to state and prove the game-theoretic channel coding theorem in Section \ref{sec:capacity}. Since our goal is to develop a non-probabilistic information theory, we will not base our machinery on probability, but rather on pricing DC cones. Recall that a pricing DC cone $A\in\mathrm{DCCone}(\mathcal{Y})\subseteq\mathbb{R}^{\mathcal{Y}}$ is a set of portfolios we can choose. It describes the following simple game: we (the encoder) choose the portfolio $a\in A$, the adversary sees $a$ and chooses the outcome $y\in\mathcal{Y}$, and then our payoff is $a(y)$.

The prototypical example of pricing DC cones are halfspaces. For a probability vector $p\in\Delta_{\mathcal{Y}}$ (for finite $\mathcal{Y}$), its \emph{halfspace DC cone} is 
\begin{equation}
p^{\circ}:=\{a\in\mathbb{R}^{\mathcal{Y}}:\langle p,a\rangle\le0\}\in\mathrm{DCCone}(\mathcal{Y}).\label{eq:halfspace}
\end{equation}
Intuitively, $a\in p^{\circ}$ means that the portfolio does not have a positive ``expected payoff'' when the outcome follows the probability distribution $p$.  There is a one-to-one correspondence between probability vectors and halfspace DC cones,\footnote{A more general correspondence between probability measures and certain convex pricing cones was noted in \cite{shafer2005probability,shafer2010base}.} though many interesting pricing DC cones are not halfspaces, such as (\ref{eq:bsc_feedback}). The notion of pricing DC cones is a strict generalization of probability vectors. 

The settings in this paper are non-probabilistic, so ``expected payoff'' has no operational meaning. For the sake of presentation, we will often use loose probabilistic analogies to explain pricing DC cones. We still avoid probability in the definitions of coding games.   

Pricing DC cones can be combined via the following operations.

\smallskip{}

\begin{defn}
\label{def:operations} For $A,B\in\mathrm{DCCone}(\mathcal{Y})$:
\begin{itemize}
\item Their \emph{sum} is $A+B$ (Minkowski sum). This means we can combine the portfolios in $A$ and $B$.
\item Their \emph{union} is $A\cup B$. This corresponds to our choice of $A$ or $B$ , i.e., we can either choose a portfolio from $A$ or from $B$, but not both.
\item Their \emph{intersection} is $A\cap B$. This corresponds to the adversary's choice of $A$ or $B$, i.e., we choose portfolios $a\in A$, $b\in B$, and upon revealing $y$, the adversary chooses the worse payoff $\min\{a(y),b(y)\}$.\footnote{Due to downward-closure, $s\in A\cap B$ if and only if there exist $a\in A$, $b\in B$ such that $s(y)=\min\{a(y),b(y)\}$.}
\end{itemize}
\end{defn}
\smallskip{}

Since DC cones can model channels, these operations can also be interpreted as operations on channels. For example, the DC cone $A$ in (\ref{eq:bsc_feedback}) can be expressed as the union of two halfspace DC cones: $A=(1-\beta,\beta)^{\circ}\cup(\beta,1-\beta)^{\circ}$, where $(1-\beta,\beta)^{\circ}$ is the halfspace DC cone of $\mathrm{Bern}(\beta)$, and $(\beta,1-\beta)^{\circ}$ is the halfspace DC cone of $\mathrm{Bern}(1-\beta)$. Hence, choosing a portfolio $a\in A$ means the encoder chooses $\mathrm{Bern}(\beta)$ or $\mathrm{Bern}(1-\beta)$, which is precisely the meaning of a $\mathrm{BSC}(\beta)$. Refer to Section \ref{sec:single_cone} for more operations.

The worst DC cone is the \emph{empty DC cone} $\emptyset\in\mathrm{DCCone}(\mathcal{Y})$, meaning that we cannot choose any portfolio. Our payoff is $-\infty$, so playing this game will bankrupt us. The second worst DC cone is the \emph{nonpositive DC cone} $\mathbb{R}_{\le0}^{\mathcal{Y}}$, meaning that we can only choose portfolios with nonpositive payoffs, so we should simply select $0\in\mathbb{R}_{\le0}^{\mathcal{Y}}$ and have zero payoff. The adversary can choose any outcome $y$. The most useful DC cone is the \emph{full DC cone} $\mathbb{R}^{\mathcal{Y}}$, meaning that we can choose any portfolio, even the portfolio $a(y)=\gamma$ for arbitrarily large $\gamma$, which guarantees a payoff $\gamma$ regardless of $y$, so we can perform arbitrage. 

The interesting DC cones are in between. For example, the \emph{noiseless DC cone} 
\[
\mathbb{R}_{\ngtr0}^{\mathcal{Y}}:=\mathbb{R}^{\mathcal{Y}}\backslash\mathbb{R}_{>0}^{\mathcal{Y}}
\]
contains all portfolios where the payoffs are not all positive,\footnote{To justify the notation $\mathbb{R}_{\ngtr0}^{\mathcal{Y}}$, note that $\mathbb{R}_{\ngtr0}^{\mathcal{Y}}=\{a\in\mathbb{R}^{\mathcal{Y}}:\,a>0\;\text{is false}\}$, where $a>0$ denotes entrywise comparison.} so no guaranteed arbitrage is possible. If we want to force the outcome to be $y=y_{0}$ for some $y_{0}\in\mathcal{Y}$, we can choose the portfolio $a\in\mathbb{R}_{\ngtr0}^{\mathcal{Y}}$, $a(y)=\gamma\mathbf{1}\{y\neq y_{0}\}$ for a large $\gamma$, so the adversary must choose $y=y_{0}$ to stop us from having an arbitrarily large payoff. Hence, the DC cone $\mathbb{R}_{\ngtr0}^{\mathcal{Y}}$ is like a ``noiseless channel'' where we can choose $y$. Refer to Figure \ref{fig:cones}.

\begin{figure*}
\centering
\includegraphics[scale=0.95]{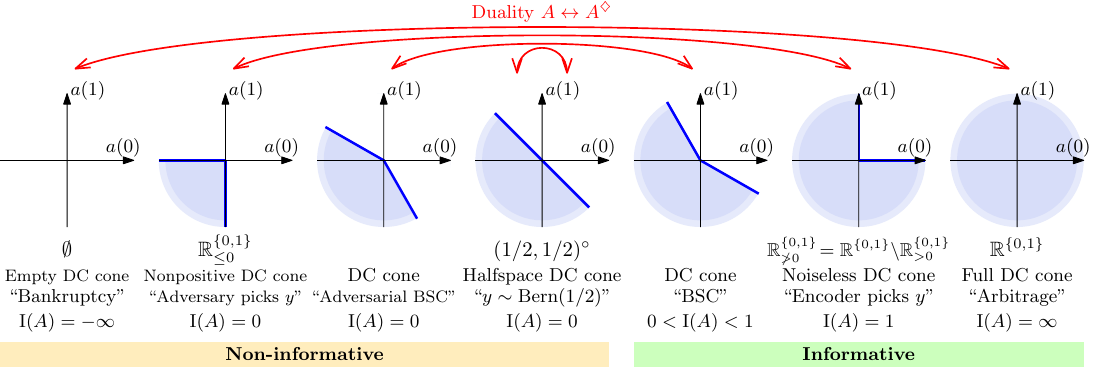}

\caption{Examples of payoff DC cones $A\in\mathrm{DCCone}(\{0,1\})$, and their information capacities $\mathrm{I}(A)$ (Definition \ref{def:info_capacity}).}\label{fig:cones}

\end{figure*}

More generally, we call $A\in\mathrm{DCCone}(\mathcal{Y})$ \emph{informative} if $\mathrm{hull}(A)=\mathbb{R}^{\mathcal{Y}}$, where $\mathrm{hull}(A)$ is the convex hull of $A$. For finite $\mathcal{Y}$, $A$ is informative if and only if $A$ is not contained in a halfspace. Halfspace DC cones and $\mathbb{R}_{\le0}^{\mathcal{Y}}$ are non-informative, whereas $\mathbb{R}^{\mathcal{Y}}$, $\mathbb{R}_{\ngtr0}^{\mathcal{Y}}$ (for $|\mathcal{Y}|\ge2$) and (\ref{eq:bsc_feedback}) are informative. Loosely speaking, halfspace DC cones do not convey information since the adversary can ``choose $y$ at random'' to guarantee nonpositive ``expected payoff'' regardless of the choice of $a\in A$. Therefore, if the encoder wants to influence $y$ by choosing a portfolio $a\in A$, $A$ cannot be a subset of a halfspace.\footnote{In contrast, a convex pricing cone \cite{shafer2005probability,shafer2010base} is convex and cannot be $\mathbb{R}^{\mathcal{Y}}$, so it cannot be informative.} More discussions on informative DC cones will be included in Section \ref{sec:capacity}.

\smallskip{}

\subsection{Duality}

A duality between DC cones is defined as follows.

\smallskip{}

\begin{defn}
\label{def:dual} For $A\in\mathrm{DCCone}(\mathcal{Y})$, its \emph{dual DC cone} is\footnote{Note that dual DC cone is not the same as the conventional concept of dual cone $A^{*}=\{b\in\mathbb{R}^{\mathcal{Y}}:\,\forall a\in A.\langle a,b\rangle\ge0\}$. The conventional dual cone is not a duality between DC cones, since the dual cone of a DC cone may not be a DC cone.}
\[
A^{\diamondsuit}:=\{\tilde{a}\in\mathbb{R}^{\mathcal{Y}}:\,\forall a\in A.\exists y\in\mathcal{Y}.\,a(y)+\tilde{a}(y)\le0\}.
\]
For a game-theoretic channel $F:\mathcal{X}\to\mathrm{DCCone}(\mathcal{Y})$, write $F^{\diamondsuit}:\mathcal{X}\to\mathrm{DCCone}(\mathcal{Y})$, $F^{\diamondsuit}(x)=(F(x))^{\diamondsuit}$.
\end{defn}
\smallskip{}

In other words, $A^{\diamondsuit}$ is the largest subset of $\mathbb{R}^{\mathcal{Y}}$ such that $A+A^{\diamondsuit}\subseteq\mathbb{R}_{\ngtr0}^{\mathcal{Y}}$. Intuitively, $A^{\diamondsuit}$ corresponds to swapping the roles of us and the adversary. Consider the game corresponding to $A$ with the roles flipped: 1) adversary chooses $a\in A$, 2) we choose $y\in\mathcal{Y}$ according to $a$, 3) the adversary gains $a(y)$ (this is a zero-sum game so we gain $-a(y)$).  We can have a payoff function $\tilde{a}\in\mathbb{R}^{\mathcal{Y}}$ if, regardless of the adversary's choice of $a$, we can select $y$ such that $\tilde{a}(y)\le-a(y)$. This gives the definition of $A^{\diamondsuit}$. Note that $\emptyset$ (empty DC cone, bankruptcy) and $\mathbb{R}^{\mathcal{Y}}$ (full DC cone, arbitrage) are duals of each other. $\mathbb{R}_{\le0}^{\mathcal{Y}}$ (adversary chooses $y$) and $\mathbb{R}_{\ngtr0}^{\mathcal{Y}}$ (we choose $y$) are duals of each other. A halfspace DC cone $p^{\circ}$ is the dual of itself.\footnote{If $\tilde{a}\in p^{\circ}$, then for any $a\in p^{\circ}$, since $\langle p,a+\tilde{a}\rangle\le0$, there exists $y$ with $a(y)+\tilde{a}(y)\le0$, so $\tilde{a}\in(p^{\circ})^{\diamondsuit}$. If $\tilde{a}\in(p^{\circ})^{\diamondsuit}$, then since $\langle p,\tilde{a}\rangle-\tilde{a}\in p^{\circ}$, there exists $y$ such that $\langle p,\tilde{a}\rangle-\tilde{a}(y)+\tilde{a}(y)\le0$, so $\tilde{a}\in p^{\circ}$.}  See Sections \ref{subsec:adv_feedback}, \ref{subsec:avcf} and \ref{sec:adv_cost} for the operational meanings of duality. We give some properties of the dual DC cone. The proof is in Appendix \ref{subsec:pf_double_dual}.

\smallskip{}

\begin{prop}
\label{prop:double_dual}Consider the box topology over $\mathbb{R}^{\mathcal{Y}}$.\footnote{The box topology is generated by the base sets $\{a\in\mathbb{R}^{\mathcal{Y}}:\forall y.\,a(y)\in S(y)\}$ for $S:\mathcal{Y}\to\mathcal{O}$, where $\mathcal{O}$ is the set of open subsets of $\mathbb{R}$.} For $A,B\in\mathrm{DCCone}(\mathcal{Y})$, we have:
\begin{itemize}
\item The closure of $A$ is $\mathrm{cl}(A)=\{c\in\mathbb{R}^{\mathcal{Y}}:\,\forall c'\in\mathbb{R}^{\mathcal{Y}}.(c'<c\,\to\,c'\in A)\}$, where $c'<c$ denotes pointwise comparison.
\item $A^{\diamondsuit}$ is closed, and $A^{\diamondsuit\diamondsuit}=\mathrm{cl}(A)$.
\item $(A\cup B)^{\diamondsuit}=A^{\diamondsuit}\cap B^{\diamondsuit}$ and $(A\cap B)^{\diamondsuit}=A^{\diamondsuit}\cup B^{\diamondsuit}$. 
\end{itemize}
\end{prop}
\smallskip{}

\begin{rem}
Let $\mathrm{CDCCone}(\mathcal{Y}):=\{A\in\mathrm{DCCone}(\mathcal{Y}):\,A\;\text{is closed}\}$ (w.r.t. the box topology). By Proposition \ref{prop:double_dual}, we can define a topological space over $\mathbb{R}^{\mathcal{Y}}$ with the collection of closed sets given by $\mathrm{CDCCone}(\mathcal{Y})$. The duality $A\mapsto A^{\diamondsuit}$ is an \emph{involutive lattice antiautomorphism} over $\mathrm{CDCCone}(\mathcal{Y})$, that is, it is involutive ($A^{\diamondsuit\diamondsuit}=A$ for $A\in\mathrm{CDCCone}(\mathcal{Y})$), and is a lattice homomorphism from the lattice $\mathrm{CDCCone}(\mathcal{Y})$ (with join $\cup$ and meet $\cap$) to the opposite lattice of $\mathrm{CDCCone}(\mathcal{Y})$ (with join $\cap$ and meet $\cup$), i.e., $(A\cup B)^{\diamondsuit}=A^{\diamondsuit}\cap B^{\diamondsuit}$ and $(A\cap B)^{\diamondsuit}=A^{\diamondsuit}\cup B^{\diamondsuit}$.
\end{rem}
\smallskip{}

\subsection{Product, Pushforward and Convex Combination}

There are two common notions of the product of probability distributions. For distributions $p_{Y},p_{Z}$ over $\mathcal{Y},\mathcal{Z}$, respectively, their product distribution $p_{Y}\times p_{Z}$ is the joint distribution of $(Y,Z)$ when $Y\sim p_{Y}$ is independent of $Z\sim p_{Z}$. For distribution $p_{Y}$ and conditional distribution $p_{Z|Y}$, their semidirect product $p_{Y}p_{Z|Y}$ is the distribution of $(Y,Z)$ when $Y\sim p_{Y}$, $Z|Y\sim p_{Z|Y}$. For pricing DC cones, there are two notions of product as well, though they do not exactly correspond to the two products of probability distributions. 

\smallskip{}

\begin{defn}
\label{def:operations-2} For $A\in\mathrm{DCCone}(\mathcal{Y})$, $B\in\mathrm{DCCone}(\mathcal{Z})$, $W:\mathcal{X}\to\mathrm{DCCone}(\mathcal{Y})$, $F:\mathcal{Y}\to\mathrm{DCCone}(\mathcal{Z})$:
\begin{itemize}
\item The \emph{product} of $A$ and $B$ is $A\otimes B\in\mathrm{DCCone}(\mathcal{Y}\times\mathcal{Z})$ defined as 
\begin{align*}
A\otimes B & :=\big\{((y,z)\mapsto a(y)+b(z))\in\mathbb{R}^{\mathcal{Y}\times\mathcal{Z}}:\\
 & \quad\quad\;a\in A,b\in B\big\}\,+\,\mathbb{R}_{\le0}^{\mathcal{Y}\times\mathcal{Z}}.
\end{align*}
This means we choose a portfolio $a\in A$ (whose payoff depends only on $y$) and a portfolio $b\in B$ (whose payoff depends on $z$), and combine $a$ and $b$.
\item The \emph{semidirect product} of $A$ and $F$ is $A\ltimes F\in\mathrm{DCCone}(\mathcal{Y}\times\mathcal{Z})$,
\begin{align*}
A\ltimes F & :=\big\{((y,z)\mapsto a(y)+f(y,z))\in\mathbb{R}^{\mathcal{Y}\times\mathcal{Z}}:\\
 & \quad\quad\,a\in A,\,f\in\mathbb{R}^{\mathcal{Y}\times\mathcal{Z}},\,\forall y.\,f(y,\cdot)\in F(y)\big\},
\end{align*}
where $f(y,\cdot)=(z\mapsto f(y,z))\in\mathbb{R}^{\mathcal{Z}}$. This means we can choose a portfolio $a\in A$, and depending on the value of $y$, choose one more portfolio $f(y,\cdot)\in F(y)$, and combine $a$ and $f(y,\cdot)$. In other words, $A\ltimes F$ describes this sequential game: 1) choose $a\in A$, 2) observe $y$, gain $a(y)$, choose $b\in F(y)$, 3) observe $z$, gain $b(z)$.
\item The \emph{semidirect product} of $A$ and $B$ is $A\ltimes B:=A\ltimes F$ where $F:\mathcal{Y}\to\mathrm{DCCone}(\mathcal{Z})$, $F(y)=B$. Note that $A\otimes B\subseteq A\ltimes B$, but the inclusion can be strict since in $A\ltimes B$, we are allowed to base the choice of the portfolio $b\in B$ on the value of $y$. 
\item Write $A^{\ltimes n}:=A\ltimes\cdots\ltimes A$ ($n$ times), and $W^{\ltimes n}:\mathcal{X}^{n}\to\mathrm{DCCone}(\mathcal{Y}^{n})$, $W^{\ltimes n}(x^{n}):=W(x_{1})\ltimes\cdots\ltimes W(x_{n})$. Similar for $A^{\otimes n}$ and $W^{\otimes n}$.
\end{itemize}
\end{defn}
\smallskip{}

We remark that the notion of product distribution corresponds to the semidirect product DC cone, not the product DC cone, i.e., for distributions $p_{Y},p_{Z}$, we have $(p_{Y}\times p_{Z})^{\circ}=p_{Y}^{\circ}\ltimes p_{Z}^{\circ}$. Loosely speaking, the product DC cone $p_{Y}^{\circ}\otimes p_{Z}^{\circ}$ corresponds to the set of couplings of $p_{Y}$ and $p_{Z}$, in the sense that the adversary can generate $(Y,Z)$ in a coupled manner as long as $Y\sim p_{Y}$ and $Z\sim p_{Z}$, and it will guarantee nonpositive ``expected payoff'' for any portfolio $a\in p_{Y}^{\circ}\otimes p_{Z}^{\circ}$.

For a probability distribution $p_{Y}$ and a function $f$, we can define the pushforward $f\#p_{Y}$ which is the distribution of $f(Y)$ when $Y\sim p_{Y}$. More generally, for a conditional distribution $p_{Z|Y}$, we can define the pushforward $p_{Z|Y}\#p_{Y}$ which is the distribution of $Z$ when $Y\sim p_{Y}$, $Z|Y\sim p_{Z|Y}$. The analogous operations for pricing DC cones are defined as follows.

\smallskip{}

\begin{defn}
\label{def:pushforward} For $A\in\mathrm{DCCone}(\mathcal{Y})$, functions $f:\mathcal{Y}\to\mathcal{Z}$ and $F:\mathcal{Y}\to\mathrm{DCCone}(\mathcal{Z})$:
\begin{itemize}
\item The \emph{pushforward} of $A$ by $f$ is $f\#A\in\mathrm{DCCone}(\mathcal{Z})$, 
\[
f\#A:=\big\{ b\in\mathbb{R}^{\mathcal{Z}}:\,(y\mapsto b(f(y)))\in A\big\}.
\]
This means the original outcome $y$ is mapped to a new outcome $z=f(y)$.  In terms of channels, this is degrading a channel by applying $f$ to its output.
\item The \emph{pushforward} of $A$ by $F$ is $F\#A\in\mathrm{DCCone}(\mathcal{Z})$, 
\[
F\#A:=\pi_{2}\#(A\ltimes F),
\]
where $\pi_{2}:\mathcal{Y}\times\mathcal{Z}\to\mathcal{Z}$, $\pi_{2}(y,z)=z$. This generalize the pushforward by a function $f:\mathcal{Y}\to\mathcal{Z}$, since $f\#A=F\#A$ when $F(y)=\mathbf{e}_{f(y)}^{\circ}$ (where $\mathbf{e}_{z}^{\circ}\in\Delta_{\mathcal{Z}}$ is the vector with entry $1$ at position $z$, and $0$ elsewhere). In terms of channels, this is analogous to stochastically degrading a channel by applying a randomized mapping to its output \cite{bergmans1973random,gallager1974capacity}.
\end{itemize}
\end{defn}
\smallskip{}

Probability distributions form a convex set, where we can take a convex combination $(1-\lambda)p_{1}+\lambda p_{2}$ of distributions $p_{1},p_{2}$. We can take convex combinations of pricing DC cones as well, though the correct way to combine DC cones is via the min-plus algebra, where we use minimum and addition instead of addition and multiplication. 

\smallskip{}

\begin{defn}
\label{def:min_plus} For $A,B\in\mathrm{DCCone}(\mathcal{Y})$ and $0\le\lambda\le1$, the \emph{min-plus combination} of $A$ and $B$ is
\[
A\stackrel{\lambda}{\oplus}B:=F\#((1-\lambda,\lambda)^{\circ})\in\mathrm{DCCone}(\mathcal{Y}),
\]
where $F:[2]\to\mathrm{DCCone}(\mathcal{Y})$, $F(1)=A$, $F(2)=B$. More explicitly,
\begin{align*}
A\stackrel{\lambda}{\oplus}B & =\big\{(y\mapsto\min\{a(y)-\lambda t,\,b(y)+(1-\lambda)t\}):\\
 & \quad\quad\,a\in A,\,b\in B,\,t\in\mathbb{R}\big\}.
\end{align*}
\end{defn}
\smallskip{}

Very loosely speaking, $F\#((1-\lambda,\lambda)^{\circ})$ means that we flip a biased coin with probability of heads being $\lambda$, and choose a portfolio in $A$ in case of tails, or a portfolio in $B$ in case of heads. We can check that the min-plus combination indeed corresponds to convex combinations of distributions, in the sense that $((1-\lambda)p_{1}+\lambda p_{2})^{\circ}=p_{1}^{\circ}\stackrel{\lambda}{\oplus}p_{2}^{\circ}$. 

The following proposition shows that duality distributes over most of these operations. The proof is in Appendix \ref{subsec:pf_dual_commute}.

\smallskip{}

\begin{prop}
\label{prop:dual_commute}For $A,B\in\mathrm{DCCone}(\mathcal{Y})$, functions $f:\mathcal{Y}\to\mathcal{Z}$ and $F:\mathcal{Y}\to\mathrm{DCCone}(\mathcal{Z})$ with finite $\mathcal{Y},\mathcal{Z}$ and $F(y)\notin\{\emptyset,\mathbb{R}^{\mathcal{Z}}\}$ for all $y$, we have $(A\ltimes F)^{\diamondsuit}=A^{\diamondsuit}\ltimes F^{\diamondsuit}$, $(A\ltimes B)^{\diamondsuit}=A^{\diamondsuit}\ltimes B^{\diamondsuit}$, $(f\#A)^{\diamondsuit}=f\#A^{\diamondsuit}$, $(F\#A)^{\diamondsuit}=F^{\diamondsuit}\#A^{\diamondsuit}$. Also, if $A,B\notin\{\emptyset,\mathbb{R}^{\mathcal{Z}}\}$, then $(A\stackrel{\lambda}{\oplus}B)^{\diamondsuit}=A^{\diamondsuit}\stackrel{\lambda}{\oplus}B^{\diamondsuit}$ for $0\le\lambda\le1$. 
\end{prop}
\smallskip{}

We now introduce three different preorders over pricing DC cones and game-theoretic channels, from the strongest to the weakest. These preorders allow us to compare which pricing DC cone can convey more information.

\smallskip{}

\begin{defn}
\label{def:degraded} For $A\in\mathrm{DCCone}(\mathcal{Y})$, $B\in\mathrm{DCCone}(\mathcal{Z})$, $V:\mathcal{X}\to\mathrm{DCCone}(\mathcal{Y})$, $W:\mathcal{X}\to\mathrm{DCCone}(\mathcal{Z})$:
\begin{itemize}
\item $B$ is \emph{dominated} by $A$ if $\mathcal{Y}=\mathcal{Z}$ and $B\subseteq A$. $W$ is dominated by $V$ (denoted as $W\subseteq V$) if $\mathcal{Y}=\mathcal{Z}$ and $W(x)\subseteq V(x)$ for all $x$. This means $A$ is preferable to $B$ since it offers more choices.
\item $B$ is\emph{ deterministically degraded} w.r.t. $A$ (denoted as $B\sqsubseteq A$) if there exists $f:\mathcal{Y}\to\mathcal{Z}$ such that $B\subseteq f\#A$. $W$ is deterministically degraded w.r.t. $V$ (denoted as $W\sqsubseteq V$) if there exists $f:\mathcal{Y}\to\mathcal{Z}$ such that for all $x$, $W(x)\subseteq f\#V(x)$. This is analogous to degrading a channel by applying a function to the output.
\item $B$ is \emph{nondeterministically degraded} w.r.t. $A$ (denoted as $B\preceq A$) if there exists $F:\mathcal{Y}\to\mathrm{DCCone}(\mathcal{Z})$ with $\mathrm{hull}(F(y))\neq\mathbb{R}^{\mathcal{Z}}$ for all $y$ (i.e., $F$ is \emph{pointwise non-informative}),  such that $B\subseteq F\#A$. $W$ is nondeterministically degraded w.r.t. $V$ (denoted as $W\preceq V$) if there exists a pointwise non-informative $F:\mathcal{Y}\to\mathrm{DCCone}(\mathcal{Z})$ such that for all $x$, $W(x)\subseteq F\#V(x)$. This is analogous to the notion of stochastic degradedness of channels \cite{bergmans1973random,gallager1974capacity} and Blackwell ordering \cite{blackwell1953equivalent}.
\end{itemize}
\end{defn}
\smallskip{}

\subsection{One-Pass Channel Coding Game}

Using the semidirect product, we can express the channel coding game in the following succinct form where the encoder makes decision in one pass. 

\smallskip{}

\begin{defn}
\label{def:game_channel_one} The \emph{one-pass channel coding game} $(W,n,L,\epsilon)$ with channel $W:\mathcal{X}\to\mathrm{DCCone}(\mathcal{Y})$, blocklength $n$, message cardinality $L\in\mathbb{Z}_{>0}$ and maximum loss $0<\epsilon<1$ is defined as follows:
\begin{itemize}
\item Adversary produces $m\in[L]$.
\item Encoder observes $m$, produces $w\in W^{\ltimes n}(\mathcal{X}^{n})$, where $W^{\ltimes n}(\mathcal{X}^{n})\in\mathrm{DCCone}(\mathcal{Y}^{n})$,
\begin{align}
W^{\ltimes n}(\mathcal{X}^{n}) & =\bigcup_{x^{n}\in\mathcal{X}^{n}}W^{\ltimes n}(x^{n})\nonumber \\
 & =\bigcup_{x^{n}\in\mathcal{X}^{n}}W(x_{1})\ltimes\cdots\ltimes W(x_{n})\label{eq:n_use}
\end{align}
is the $n$\emph{-use DC cone}.
\item Adversary observes $w$, produces $y^{n}\in\mathcal{Y}^{n}$.
\item Decoder observes $y^{n}$, produces $\hat{m}\in[L]$.
\item The team wins if $w(y^{n})-\mathbf{1}\{\hat{m}\neq m\}\ge-\epsilon$.
\end{itemize}
\end{defn}
\smallskip{}

To see why Definition \ref{def:game_channel_one} is equivalent to Definition \ref{def:game_channel}, note that choosing $w\in W(x_{1})\ltimes\cdots\ltimes W(x_{n})$ allows the encoder to choose portfolios $a_{i}\in W(x_{i})$ for $i\in[n]$, where the choice of $a_{i}$ can depend on $y^{i-1}$. By choosing $w$, the encoder designs a ``rule-based investing strategy'', which includes the whole decision tree of which portfolio to choose upon which values of $y_{i}$. This strategy replaces the interactive steps in Definition \ref{def:game_channel}.\footnote{It might appear that Definition \ref{def:game_channel_one} has a stronger adversary than Definition \ref{def:game_channel} since Definition \ref{def:game_channel_one} requires the encoder to announce the whole strategy $w$ to the adversary. Nevertheless, this is inconsequential to the definition of a game being a guaranteed win, which always allows the adversary's strategy to depend on the team's strategy.}  Using the channel $W^{\ltimes n}(\mathcal{X}^{n})$ once is equivalent to $n$ uses of the channel $W$. The DC cone $W^{\ltimes n}(\mathcal{X}^{n})$ is usually non-convex and informative.  

We can write the channel coding game in an even more succinct manner using deterministic degradedness (Definition \ref{def:degraded}): the channel coding game is a guaranteed win if and only if 
\begin{equation}
\mathbb{R}_{\ngtr0}^{[L]}\stackrel{\epsilon}{\oplus}\mathbb{R}_{\le0}^{[L]}\,\sqsubseteq\,W^{\ltimes n}(\mathcal{X}^{n}).\label{eq:channel_require}
\end{equation}
To understand this expression, $\mathbb{R}_{\ngtr0}^{[L]}$ is the noiseless DC cone which describes the scenario where the encoder can send $m\in[L]$ noiselessly. $\mathbb{R}_{\le0}^{[L]}$ is the nonpositive DC cone (the dual of $\mathbb{R}_{\ngtr0}^{[L]}$) which describes the scenario where the encoder fails to convey any information. The DC cone $\mathbb{R}_{\ngtr0}^{[L]}\stackrel{\epsilon}{\oplus}\mathbb{R}_{\le0}^{[L]}$ represents the decoding requirement. It is a mixture of the noiseless DC cone (with mixture weight $1-\epsilon$, representing success) and the nonpositive DC cone (with weight $\epsilon$, representing failure).  The deterministic degradedness constraint in (\ref{eq:channel_require}) means that we can simulate the channel $\mathbb{R}_{\ngtr0}^{[L]}\stackrel{\epsilon}{\oplus}\mathbb{R}_{\le0}^{[L]}$ by applying a function to the output of the channel $W^{\ltimes n}(\mathcal{X}^{n})$.  Direct computation gives
\[
\mathbb{R}_{\ngtr0}^{[L]}\stackrel{\epsilon}{\oplus}\mathbb{R}_{\le0}^{[L]}=\big\{ a\in\mathbb{R}^{[L]}:\,\epsilon\max_{m}a(m)+(1-\epsilon)\min_{m}a(m)\le0\big\}.
\]

Pricing DC cones unify the concepts of channels and decoding requirements. Expressing channel coding as a degradedness constraint would not be possible using probabilistic channels.\footnote{Decoding is successful with probability $1-\epsilon$ does not mean that the $L$-ary symmetric channel from $M\in[L]$ to $\hat{M}\in[L]$ with $p(m|m)=1-\epsilon$ and $p(\hat{m}|m)=\epsilon/(L-1)$ for $\hat{m}\neq m$ is degraded compared to the $n$-fold channel $W^{n}$, since error may not occur in a symmetric manner.} This is only possible in (\ref{eq:channel_require}) due to the generality of pricing DC cones. 

\smallskip{}

\section{Capacity Theorem for the Game-Theoretic Channel}\label{sec:capacity}

Recall that a pricing DC cone $A$ is informative if $\mathrm{hull}(A)=\mathbb{R}^{\mathcal{Y}}$. A more informative DC cone allows the encoder to influence $y$ by choosing the portfolio $a\in A$. Consider the following informally defined game: 1) the adversary chooses a ``prior distribution'' $q\in\Delta_{\mathcal{Y}}$; 2) the encoder chooses $a\in A$; 3) the adversary chooses a ``posterior distribution'' $p\in\Delta_{\mathcal{Y}}$ according to $a$ subject to $\langle p,a\rangle\le0$ (the expected payoff when $y\sim p$ is nonpositive); 4) the adversary succeeds if $p\approx q$ (informally). If the adversary succeeds in making $p=q$ (instead of only $p\approx q$), then $A$ is non-informative, since the adversary can choose $q$ without seeing $a$, so the choice of $a$ has no influence on the output. An analogy in statistics is that if the posterior distribution is the same as the prior distribution, then the data gives no new information. For a more informative $A$, the adversary must choose $p$ according to $a$, and $p$ cannot always be close to $q$. This motivates the following measure of informativeness, where ``$p\approx q$'' is measured by the KL divergence.

\smallskip{}

\begin{defn}
\label{def:info_capacity} The \emph{information capacity} of a pricing DC cone $A\in\mathrm{DCCone}(\mathcal{Y})$  is
\[
\mathrm{I}(A):=\inf_{q\in\Delta_{\mathcal{Y}}}\,\sup_{a\in A}\,\inf_{p\in\Delta_{\mathcal{Y}}:\,\langle p,a\rangle\le0}D(p\Vert q),
\]
where $\Delta_{\mathcal{Y}}$ is the set of finite distributions over $\mathcal{Y}$, and $D(p\Vert q):=\sum_{y\in\mathrm{supp}(p)}p(y)\log(p(y)/q(y))$ is the KL divergence (which is $\infty$ if $\mathrm{supp}(p)\nsubseteq\mathrm{supp}(q)$). We take $\mathrm{I}(\emptyset)=-\infty$ and $\mathrm{I}(\mathbb{R}^{\mathcal{Y}})=\infty$. 
\end{defn}
\smallskip{}

We can now state the main result of this paper, which shows that the capacity of the game-theoretic channel (Definition \ref{def:capacity}) is upper-bounded by, and is sometimes equal to, the information capacity of the range DC cone $W(\mathcal{X})$. 

\smallskip{}

\begin{thm}
[Game-theoretic channel coding theorem]\label{thm:cap}The capacity of the game-theoretic channel $W:\mathcal{X}\to\mathrm{DCCone}(\mathcal{Y})$ (where $\mathcal{X},\mathcal{Y}$ are finite) satisfies
\[
\mathrm{C}(W)\le\mathrm{I}(W(\mathcal{X})),
\]
where $W(\mathcal{X})=\bigcup_{x\in\mathcal{X}}W(x)$. Equality holds if $\mathbb{R}^{\mathcal{Y}}\backslash W(x)$ is open and convex for all $x$ (we call such $W$ \emph{pointwise closed and concave}). 
\end{thm}
\smallskip{}

 Theorem \ref{thm:cap} implies the conventional channel coding theorem with or without feedback (Sections \ref{subsec:dmc}, \ref{subsec:dmc_feedback}). For other settings where $W$ may not be pointwise closed and concave, the theorem only gives an upper bound on the capacity. Nevertheless, considering that game-theoretic channel coding generalizes zero-error channel coding (Section \ref{subsec:adv}) which is a major open problem, finding $\mathrm{C}(W)$ for all $W$ would be an even harder problem. 

The achievability proof of Theorem \ref{thm:cap} is deferred to Section \ref{subsec:avcf}. The remainder of this section is dedicated to the converse proof, which is based on the following properties of the information capacity of pricing DC cones. The proofs of these properties are given in Appendix \ref{subsec:pf_i_prop}.

\smallskip{}

\begin{prop}
\label{prop:i_prop}For nonempty $A\in\mathrm{DCCone}(\mathcal{Y})$, $B\in\mathrm{DCCone}(\mathcal{Z})$ where $\mathcal{Y},\mathcal{Z}$ are finite, we have:
\begin{itemize}
\item $\mathrm{I}(A)=0$ if and only if $\mathrm{hull}(A)\neq\mathbb{R}^{\mathcal{Y}}$, i.e., $A$ is non-informative. 
\item (Monotonicity) If $B\preceq A$ (Definition \ref{def:degraded}), then $\mathrm{I}(B)\le\mathrm{I}(A)$.
\item (Additivity) 
\[
\mathrm{I}(A\otimes B)=\mathrm{I}(A\ltimes B)=\mathrm{I}(A)+\mathrm{I}(B).
\]
\item (Alternative characterization) 
\begin{equation}
\mathrm{I}(A)=\sup_{p_{T}}\inf_{p_{Y|T}}I(T;Y),\label{eq:alt_char}
\end{equation}
where $I(T;Y)$ is the mutual information, the supremum is over finite discrete distributions $p_{T}$ over $A$, and the infimum is over conditional distributions $p_{Y|T}$ from $A$ to $\mathcal{Y}$ satisfying that $\langle t,p_{Y|T}(\cdot|t)\rangle\le0$ for every $t\in A$. The infimum is $\infty$ if no such $p_{Y|T}$ exists. 
\item (Generalizing probabilistic Shannon capacity) If $A=\bigcup_{x}(p_{Y|X}(\cdot|x))^{\circ}$, where $p_{Y|X}$ is a conditional distribution from a finite $\mathcal{X}$ to a finite $\mathcal{Y}$, then $\mathrm{I}(A)=\max_{p_{X}}I(X;Y)$ is the capacity of the channel $p_{Y|X}$. 
\end{itemize}
\end{prop}
\smallskip{}

The alternative characterization (\ref{eq:alt_char}) gives an interpretation of the upper bound in Theorem \ref{thm:cap} using more familiar probabilistic notions. Intuitively, the encoder chooses $p_{T}$ (a distribution of the choice of portfolio $t\in A$), and the adversary controls $p_{Y|T}$ (how $y$ is chosen based on $t$). The adversary must ensure $\sum_{y}t(y)p_{Y|T}(y|t)\le0$, i.e., the expected payoff when the encoder chooses the portfolio $t\in A$ is nonpositive. Hence, $\mathrm{IC}(W(\mathcal{X}))$ is the minimum capacity among all ``no expected profit'' channels $p_{Y|T}$. This is similar to the arbitrarily varying channel \cite{blackwell1959capacity,ahlswede1978elimination,lapidoth2002reliable,csiszar2002capacity}. Nevertheless, this is not a valid converse proof of Theorem \ref{thm:cap} since game-theoretic channel coding is non-probabilistic, so ``expected profit'' is undefined.  We now give a converse proof of Theorem \ref{thm:cap} using Proposition \ref{prop:i_prop}.
\begin{IEEEproof}
[Upper bound in Theorem \ref{thm:cap}]Assume the channel coding game is a guaranteed win. By (\ref{eq:channel_require}), $\mathbb{R}_{\ngtr0}^{[L]}\stackrel{\epsilon}{\oplus}\mathbb{R}_{\le0}^{[L]}\sqsubseteq W^{\ltimes n}(\mathcal{X}^{n})$. By Proposition \ref{prop:i_prop},
\begin{align*}
 & n\mathrm{I}(W(\mathcal{X}))\\
 & \stackrel{(a)}{=}\mathrm{I}((W(\mathcal{X}))^{\ltimes n})\\
 & \stackrel{(b)}{\ge}\mathrm{I}(W^{\ltimes n}(\mathcal{X}^{n}))\\
 & \stackrel{(c)}{\ge}\mathrm{I}(\mathbb{R}_{\ngtr0}^{[L]}\stackrel{\epsilon}{\oplus}\mathbb{R}_{\le0}^{[L]})\\
 & \stackrel{(d)}{=}\inf_{q\in\Delta_{[L]}}\,\sup_{a\in\{(\hat{m}\mapsto\mathbf{1}\{\hat{m}\neq m\}-\epsilon):m\in[L]\}}\,\inf_{p\in\Delta_{[L]}:\langle a,p\rangle\le0}D(p\Vert q)\\
 & =\inf_{q\in\Delta_{[L]}}\,\max_{m\in[L]}\,\inf_{p\in\Delta_{[L]}:\,p(m)\ge1-\epsilon}D(p\Vert q)\\
 & \stackrel{(e)}{=}\inf_{p\in\Delta_{[L]}:\,p(1)\ge1-\epsilon}D(p\Vert u)\\
 & =\log L-H_{\mathrm{b}}(\epsilon)-\epsilon\log(L-1)\\
 & \ge(1-\epsilon)\log L-1,
\end{align*}
where (a) is by additivity, (b) is by monotonicity and $W(x_{1})\ltimes\cdots\ltimes W(x_{n})\subseteq(W(\mathcal{X}))^{\ltimes n}$, (c) is by monotonicity, (d) is because
\begin{align*}
 & \mathbb{R}_{\ngtr0}^{[L]}\stackrel{\epsilon}{\oplus}\mathbb{R}_{\le0}^{[L]}\\
 & =\big\{(\hat{m}\mapsto t(\mathbf{1}\{\hat{m}\neq m\}-\epsilon)):m\in[L],t\ge0\big\}+\mathbb{R}_{\le0}^{[L]}
\end{align*}
by direct evaluation, and (e) is by the convexity of $D(p\Vert q)$ and symmetry, where $u\in\Delta_{[L]}$, $u(m)=1/L$, and $H_{\mathrm{b}}$ is the binary entropy function. The above is essentially a non-probabilistic Fano's inequality \cite{fano1961transmission}. Hence, $(\log L)/n\le(1-\epsilon)^{-1}(\mathrm{I}(W(\mathcal{X}))+1/n)$. Taking $\epsilon\to0$, $n\to\infty$ gives $\mathrm{C}(W)\le\mathrm{I}(W(\mathcal{X}))$.  The achievability proof is deferred to Section \ref{subsec:avcf}.
\end{IEEEproof}
\smallskip{}

\section{Special Cases of Game-Theoretic Channels}\label{sec:cases}

In this section, we will study various conventional notions of channels, namely adversarial channels and discrete memoryless channels with or without feedback, and see how the framework of game-theoretic channels and Theorem \ref{thm:cap} apply to these channels.

\subsection{Zero-Error Channel Coding}\label{subsec:adv}

A zero-error channel \cite{shannon1956zero} is given by a characteristic bipartite graph $\mathcal{E}\subseteq\mathcal{X}\times\mathcal{Y}$, which means that when the input is $x\in\mathcal{Y}$, the adversary can choose any output $y\in\mathcal{Y}$ with $(x,y)\in\mathcal{E}$.  A zero-error coding scheme is given by an encoding function $f_{n}:[L]\to\mathcal{X}^{n}$ and a decoding function $g_{n}:\mathcal{Y}^{n}\to[L]$ such that for every $m\in[L]$ and $y^{n}$ with $(x_{i},y_{i})\in\mathcal{E}$ for all $i\in[n]$ where $x^{n}=f_{n}(m)$, we must have $g_{n}(y^{n})=m$. We now construct a corresponding game-theoretic channel, and show that the channel coding game is a win (regardless of the maximum loss  $0<\epsilon<1$) if and only if zero-error coding is possible.

\smallskip{}

\begin{thm}
\label{thm:adv}For finite $\mathcal{X},\mathcal{Y}$, the channel coding game $(W_{\mathrm{ad}},n,L,\epsilon)$ (Definition \ref{def:game_channel} or \ref{def:game_channel_one}) for $W_{\mathrm{ad}}:\mathcal{X}\to\mathrm{DCCone}(\mathcal{Y})$,
\[
W_{\mathrm{ad}}(x):=\bigcap_{y:\,(x,y)\in\mathcal{E}}\mathbf{e}_{y}^{\circ},
\]
(where $\mathbf{e}_{y}\in\mathbb{R}^{\mathcal{Y}}$, $\mathbf{e}_{y}(y')=\mathbf{1}\{y'=y\}$, and the above intersection is $\mathbb{R}^{\mathcal{Y}}$ if there is no such $y$) and $0<\epsilon<1$ is a guaranteed win if and only if a zero-error coding scheme exists with blocklength $n$ and message cardinality $L$. Hence, $\mathrm{C}(W_{\mathrm{ad}})$ equals the zero-error capacity of $\mathcal{E}$.
\end{thm}
\smallskip{}

\begin{IEEEproof}
The assumption that ``for every $x$, there exists $y$ such that $(x,y)\in\mathcal{E}$'' is usually imposed on $\mathcal{E}$, though we do not require this assumption since the theorem is trivially true in this case. If there is an $x$ with no possible output $y$, then any coding scheme that always uses this $x$ is vacuously valid, and $W_{\mathrm{ad}}(x)=\mathbb{R}^{\mathcal{Y}}$ so the team can perform arbitrage. The capacity is infinite. Hence, we focus on the case where for every $x$, there exists $y$ such that $(x,y)\in\mathcal{E}$.

Intuitively, portfolios in $\bigcap_{y:\,(x,y)\in\mathcal{E}}\mathbf{e}_{y}^{\circ}$ must have nonpositive payoff when $(x,y)\in\mathcal{E}$, but can have arbitrarily large payoff when $(x,y)\notin\mathcal{E}$, so the adversary must ensure $(x_{i},y_{i})\in\mathcal{E}$ to stop the encoder from winning. We now give a rigorous proof. For the ``if'' direction, consider a coding scheme $f_{n},g_{n}$. We now describe a strategy for the channel coding game. The encoder produces $x^{n}=f_{n}(m)$. At time $i$, the encoder chooses the portfolio $w_{i}\in W_{\mathrm{ad}}(x_{i})$, $w_{i}(y):=\mathbf{1}\{(x_{i},y)\notin\mathcal{E}\}$. The decoder outputs $\hat{m}=g_{n}(y^{n})$. If the adversary chooses $y_{i}$ with $(x_{i},y_{i})\notin\mathcal{E}$, then $w_{i}(x_{i},y_{i})=1$ is sufficient to win the game. If the adversary never chooses $y_{i}$ with $(x_{i},y_{i})\notin\mathcal{E}$, then $g_{n}(y^{n})=m$ so the team also wins. 

For the ``only if'' direction, consider a strategy for the channel coding game which guarantees a win. Let $f_{n}(m)$ be the $x^{n}$ produced the encoder upon the message $m$, and $g_{n}(y^{n})$ be the $\hat{m}$ produced by the decoder upon $y^{n}$. Assume the contrary that there exists $y^{n}$ with $(x_{i},y_{i})\in\mathcal{E}$ for all $i\in[n]$ where $g_{n}(y^{n})\neq m$. Since the team wins if the adversary outputs $y^{n}$, we have $\sum_{i}w_{i}(y_{i})\ge1$. However, since $w_{i}\in W_{\mathrm{ad}}(x_{i})\subseteq\mathbf{e}_{y_{i}}^{\circ}$, we must have $w_{i}(y_{i})\le0$, leading to a contradiction. Hence, $f_{n},g_{n}$ is a valid coding scheme.
\end{IEEEproof}

\smallskip{}

Note that Theorem \ref{thm:cap} gives the bound
\begin{equation}
\mathrm{C}(W)\le\mathrm{I}(W_{\mathrm{ad}}(\mathcal{X}))=\sup_{P_{X}}\inf_{P_{Y|X}:\mathbb{P}((X,Y)\in\mathcal{E})=1}I(X;Y).\label{eq:adv_ub}
\end{equation}
This is the upper bound on the zero-error capacity in \cite{shannon1956zero}, which is generally loose.

\smallskip{}

\subsection{Zero-Error Channels with Feedback}\label{subsec:adv_feedback}

A zero-error channel with feedback \cite{shannon1956zero} is similar to the zero-error channel in Section \ref{subsec:adv}, except that the encoder can observe the past channel outputs $y^{i-1}$ in order to produce $x_{i}$ at time $i$. A zero-error feedback coding scheme is given by the encoding functions $f_{n,i}:[L]\times\mathcal{Y}^{i-1}\to\mathcal{X}$ for $i\in[n]$ and decoding function $g_{n}:\mathcal{Y}^{n}\to[L]$, such that for every $m\in[L]$ and $x^{n},y^{n}$ with $x_{i}=f_{n,i}(m,y^{i-1})$ and $(x_{i},y_{i})\in\mathcal{E}$ for all $i\in[n]$, we must have $g_{n}(y^{n})=m$. We now construct a game-theoretic channel, and show that the channel coding game is a win if and only if zero-error coding is possible. Note that this game-theoretic channel has a singleton input set $\{0\}$, i.e., it has no explicit input. Instead, the encoder influence $y_{i}$ through the choice of portfolio.

\smallskip{}

\begin{thm}
\label{thm:adv_feedback}For finite $\mathcal{X},\mathcal{Y}$, the channel coding game $(W_{\mathrm{adfb}},n,L,\epsilon)$ (Definition \ref{def:game_channel} or \ref{def:game_channel_one}) for $W_{\mathrm{adfb}}:\{0\}\to\mathrm{DCCone}(\mathcal{Y})$,
\[
W_{\mathrm{adfb}}(0):=\bigcup_{x\in\mathcal{X}}\,\bigcap_{y:\,(x,y)\in\mathcal{E}}\mathbf{e}_{y}^{\circ},
\]
(where $\mathbf{e}_{y}\in\mathbb{R}^{\mathcal{Y}}$, $\mathbf{e}_{y}(y')=\mathbf{1}\{y'=y\}$) is a guaranteed win if and only if a zero-error feedback coding scheme exists with blocklength $n$ and message cardinality $L$. Hence, $\mathrm{C}(W_{\mathrm{adfb}})$ equals the zero-error feedback capacity of $\mathcal{E}$.
\end{thm}
\smallskip{}

\begin{IEEEproof}
For the ``if'' direction, consider a coding scheme $f_{n,i},g_{n}$. We now describe a strategy for the channel coding game. At time $i$, the encoder computes $\tilde{x}_{i}=f_{n,i}(m,y^{i-1})$, and chooses the portfolio $w_{i}\in W_{\mathrm{adfb}}(0)$, $w_{i}(y):=\mathbf{1}\{(\tilde{x}_{i},y)\notin\mathcal{E}\}$. The decoder outputs $\hat{m}=g_{n}(y^{n})$. If the adversary chooses $y_{i}$ with $(\tilde{x}_{i},y_{i})\notin\mathcal{E}$, then $w_{i}(\tilde{x}_{i},y_{i})=1$ is sufficient to win the game. If the adversary never chooses $y_{i}$ with $(\tilde{x}_{i},y_{i})\notin\mathcal{E}$, then $g_{n}(y^{n})=m$ so the team also wins. 

For the ``only if'' direction, consider a strategy for the channel coding game which guarantees a win. Let $f_{n,i}(m,y^{i-1})$ be an $\tilde{x}\in\mathcal{X}$ such that $w_{i}(m,y^{i-1})\in\bigcap_{y:\,(\tilde{x},y)\in\mathcal{E}}\mathbf{e}_{y}^{\circ}$, where $w_{i}(m,y^{i-1})$ is the output of the encoder at time $i$ upon the message $m$ and past sequence $y^{i-1}$ (such $x$ exists since $w_{i}(m,y^{i-1})\in W_{\mathrm{adfb}}(0)$), and $g_{n}(y^{n})$ be the $\hat{m}$ produced by the decoder upon $y^{n}$. Assume the contrary that there exists $\tilde{x}^{n},y^{n}$ with $\tilde{x}_{i}=f_{n,i}(m,y^{i-1})$ and $(\tilde{x}_{i},y_{i})\in\mathcal{E}$ for all $i\in[n]$, but $g_{n}(y^{n})\neq m$. Since the team wins if the adversary outputs $y^{n}$, we have $\sum_{i}w_{i}(y_{i})\ge1$ where $w_{i}=w_{i}(m,y^{i-1})$. However, since $w_{i}\in\bigcap_{y:\,(\tilde{x}_{i},y)\in\mathcal{E}}\mathbf{e}_{y}^{\circ}\subseteq\mathbf{e}_{y_{i}}^{\circ}$, we must have $w_{i}(y_{i})\le0$, leading to a contradiction. Hence, $f_{n,i},g_{n}$ is a valid coding scheme.
\end{IEEEproof}

\smallskip{}

Comparing $W_{\mathrm{adfb}}$ with $W_{\mathrm{ad}}$ in Theorem \ref{thm:adv}, we can see that  $W_{\mathrm{adfb}}(0)=W_{\mathrm{ad}}(\mathcal{X})$, i.e., all the range of $W_{\mathrm{ad}}$ is included in $W_{\mathrm{adfb}}(0)$. The difference between having an explicit input (in $W_{\mathrm{ad}}$) and combining all inputs as a union (in $W_{\mathrm{adfb}}$) is that an explicit input $x^{n}$ must be fixed at the beginning of the game in Definition \ref{def:game_channel}, whereas the choice of portfolio $w_{i}$ can depend on the past $y^{i-1}$. The zero-error feedback capacity is generally larger than the zero-error non-feedback capacity \cite{shannon1956zero}. Hence, even though $W_{\mathrm{adfb}}$ and $W_{\mathrm{ad}}$ have the same range DC cone, they do not generally have the same capacity. 

We may consider a ``dual channel'' where the roles of the encoder and the adversary are swapped. At time $i$, the adversary produces $x_{i}$, and then the encoder chooses $y_{i}$ with $(x_{i},y_{i})\in\mathcal{E}$. This corresponds to the game-theoretic channel $W_{\mathrm{adfb}}^{\diamondsuit}:\{0\}\to\mathrm{DCCone}(\mathcal{Y})$, where $W_{\mathrm{adfb}}^{\diamondsuit}(0)=(W_{\mathrm{adfb}}(0))^{\diamondsuit}$ is the dual DC cone (Definition \ref{def:dual}). The dual channel can be computed as
\begin{align*}
W_{\mathrm{adfb}}^{\diamondsuit}(0) & =\bigcap_{x\in\mathcal{X}}\,\bigcup_{y:\,(x,y)\in\mathcal{E}}\mathbf{e}_{y}^{\circ}\\
 & =\bigcup_{f\in\mathcal{Y}^{\mathcal{X}}:\forall x.(x,f(x))\in\mathcal{E}}\,\bigcap_{x\in\mathcal{X}}\mathbf{e}_{f(x)}^{\circ}\\
 & =\bigcup_{\mathcal{S}\subseteq\mathcal{Y}:\,\forall x\in\mathcal{X}.\exists y\in\mathcal{S}.(x,y)\in\mathcal{E}}\,\bigcap_{y\in\mathcal{S}}\mathbf{e}_{y}^{\circ}.
\end{align*}
Comparing this with Theorem \ref{thm:adv_feedback}, we can see that the dual channel is the zero-error channel with a characteristic bipartite graph where the left vertex set is $\{\mathcal{S}\subseteq\mathcal{Y}:\,\forall x\in\mathcal{X}.\exists y\in\mathcal{S}.(x,y)\in\mathcal{E}\}$ (the sets $\mathcal{S}\subseteq\mathcal{Y}$ that covers $\mathcal{X}$ in $\mathcal{E}$), the right vertex set is $\mathcal{Y}$, and there is an edge $(\mathcal{S},y)$ if $y\in\mathcal{S}$.

\smallskip{}

\subsection{Discrete Memoryless Channels without Feedback}\label{subsec:dmc}

Consider a discrete memoryless channel (DMC) $p_{Y|X}$ from $\mathcal{X}$ to $\mathcal{Y}$. A coding scheme with blocklength $n$, message cardinality $L$ and maximum error probability $\epsilon$ is given by an encoding function $f_{n}:[L]\to\mathcal{X}^{n}$ and a decoding function $g_{n}:\mathcal{Y}^{n}\to[L]$ such that when $M\in[L]$, $X^{n}=f_{n}(M)$, $Y_{i}|X_{i}\sim p_{Y|X}$, we have $\max_{m\in[L]}\mathbb{P}(M\neq g_{n}(Y^{n})|M=m)\le\epsilon$. 

Recall that the distribution $p_{Y|X}(\cdot|x)$ corresponds to the halfspace DC cone $(p_{Y|X}(\cdot|x))^{\circ}=\{a\in\mathbb{R}^{\mathcal{Y}}:\,\langle p_{Y|X}(\cdot|x),a\rangle\le0\}$. Hence, a corresponding game-theoretic channel can be constructed as $W(x)=(p_{Y|X}(\cdot|x))^{\circ}$. The proof of the following result is similar to Proposition \ref{prop:bsc}, though $x$ is considered an explicit input here since there is no feedback. For the proof of Theorem \ref{thm:dmc}, refer to Theorem \ref{thm:avcf} which subsumes Theorem \ref{thm:dmc}.

\smallskip{}

\begin{thm}
\label{thm:dmc}For finite $\mathcal{X},\mathcal{Y}$, the channel coding game $(W_{\mathrm{dmc}},n,L,\epsilon)$ (Definition \ref{def:game_channel} or \ref{def:game_channel_one}) for $W_{\mathrm{dmc}}:\mathcal{X}\to\mathrm{DCCone}(\mathcal{Y})$, $W_{\mathrm{dmc}}(x)=(p_{Y|X}(\cdot|x))^{\circ}$ is a guaranteed win if and only if a coding scheme exists for the DMC $p_{Y|X}$ with blocklength $n$, message cardinality $L$ and maximum error probability $\epsilon$. Hence, $\mathrm{C}(W_{\mathrm{dmc}})$ equals the probabilistic Shannon capacity of $p_{Y|X}$.
\end{thm}
\smallskip{}

Although $W_{\mathrm{dmc}}(x)$ is convex, the $n$-use cone $W_{\mathrm{dmc}}^{\ltimes n}(\mathcal{X}^{n})$ (\ref{eq:n_use}) is non-convex. Also, note that $\mathrm{I}(W_{\mathrm{dmc}}(\mathcal{X}))$ is the probabilistic Shannon capacity of $p_{Y|X}$ due to Proposition \ref{prop:i_prop}. Hence, Theorem \ref{thm:cap} also implies that $\mathrm{C}(W_{\mathrm{dmc}})=\mathrm{I}(W_{\mathrm{dmc}}(\mathcal{X}))$ is the probabilistic Shannon capacity.\smallskip{}

\subsection{Discrete Memoryless Channels with Feedback}\label{subsec:dmc_feedback}

We now consider the DMC $p_{Y|X}$ with feedback. The meaning of a feedback coding scheme has been explained in Section \ref{sec:martingale}, so it will not be repeated here. We have seen in Section \ref{subsec:adv_feedback} that, to add feedback to a channel, we can take the union of the DC cones $W_{\mathrm{dmc}}(x)$ to form $W_{\mathrm{dmcfb}}(0)=W_{\mathrm{dmc}}(\mathcal{X})=\bigcup_{x\in\mathcal{X}}W_{\mathrm{dmc}}(x)$. For the proof of Theorem \ref{thm:dmc_feedback}, refer to Theorem \ref{thm:avcf} which subsumes Theorem \ref{thm:dmc_feedback}.

\smallskip{}

\begin{thm}
\label{thm:dmc_feedback}For finite $\mathcal{X},\mathcal{Y}$, the channel coding game $(W_{\mathrm{dmcfb}},n,L,\epsilon)$ (Definition \ref{def:game_channel} or \ref{def:game_channel_one}) for $W_{\mathrm{dmcfb}}:\{0\}\to\mathrm{DCCone}(\mathcal{Y})$, $W_{\mathrm{dmcfb}}(0)=\bigcup_{x\in\mathcal{X}}(p_{Y|X}(\cdot|x))^{\circ}$ is a guaranteed win if and only if a feedback coding scheme exists for the DMC $p_{Y|X}$ with blocklength $n$, message cardinality $L$ and maximum error probability $\epsilon$. Hence, $\mathrm{C}(W_{\mathrm{dmcfb}})$ equals the Shannon capacity of $p_{Y|X}$.
\end{thm}
\smallskip{}

Theorem \ref{thm:cap} also implies that $\mathrm{C}(W_{\mathrm{dmcfb}})=\mathrm{I}(W_{\mathrm{dmcfb}}(\mathcal{X}))$ is the Shannon capacity. Channels with noiseless feedback are characterized by one pricing DC cone, and do not require an explicit input $x$, as its input can be incorporated into the pricing DC cone. In this sense, channels with noiseless feedback are more natural than channels without feedback. Channels characterized by one DC cone (\emph{single-cone channels}) will be further discussed in Section \ref{sec:single_cone}. 

\smallskip{}

\subsection{Arbitrarily Varying Channels with Feedback}\label{subsec:avcf}

In this section, we study arbitrarily varying channels with feedback (AVCF), which is a general class of channels (subsuming the channels in Sections \ref{subsec:adv}, \ref{subsec:adv_feedback}, \ref{subsec:dmc}, \ref{subsec:dmc_feedback}) that can be modelled as a game-theoretic channel. Unlike usual AVCs \cite{blackwell1959capacity,ahlswede1978elimination,lapidoth2002reliable,csiszar2002capacity}, we assume the existence of noiseless feedback, so the encoder and the adversary know $Y^{i-1}$ at time $i$.\footnote{Another notion of AVC with feedback was studied in \cite{lomnitz2013universal}. Unlike \cite{lomnitz2013universal}, we split the input into a causal part and a noncausal part.} The input to the channel at time $i$ is a pair $(X_{i},Z_{i})\in\mathcal{X}\times\mathcal{Z}$, where $X_{1},\ldots,X_{n}$ (the \emph{noncausal inputs}) must be chosen by the encoder at the beginning and cannot depend on the feedback, whereas $Z_{i}$ (the \emph{causal inputs}) can depend on $Y^{i-1}$. The adversary produces $V_{i}$ (the \emph{adversarial input}) at time $i$, which can depend on $Y^{i-1}$ and $Z^{i}$. This aspect is similar to the causal adversarial channel \cite{chen2015characterization}.

\smallskip{}

\begin{defn}
A\emph{ arbitrarily varying channel with feedback} (AVCF) is characterized by a conditional distribution $p_{Y|X,Z,V}$ from $\mathcal{X}\times\mathcal{Z}\times\mathcal{V}$ to $\mathcal{Y}$, where $\mathcal{X}$ is the noncausal input set, $\mathcal{Z}$ is the causal input set, $\mathcal{V}$ is the adversary's action set, and $\mathcal{Y}$ is the output set. Operationally, for a blocklength $n$ and message cardinality $L$:
\begin{itemize}
\item Encoder observes message $M\in[L]$, produces $X^{n}:=x^{n}(M)\in\mathcal{X}^{n}$.
\item For $i=1,\ldots,n$:
\begin{itemize}
\item Encoder produces $Z_{i}:=z_{i}(m,Y^{i-1})\in\mathcal{Z}$.
\item Adversary produces  $V_{i}:=v_{i}(m,Y^{i-1},Z^{i})\in\mathcal{V}$. 
\item Generate $Y_{i}|(X_{i},Z_{i},V_{i})\sim P_{Y|X,Z,V}$.
\end{itemize}
\item Decoder produces $\hat{M}:=\hat{m}(Y^{n})\in[L]$.
\end{itemize}
A coding scheme with maximum error probability $\epsilon$ is given by the tuple of functions $(x^{n},(z_{i})_{i\in[n]},\hat{m})$, satisfying that
\[
\max_{m\in[L],\,(v_{i})_{i}}\mathbb{P}(\hat{M}\neq m\,|\,M=m)\le\epsilon,
\]
i.e., we consider the worst $m$ and the adversary's strategy $(v_{i})_{i\in[n]}$ (where $v_{i}:[L]\times\mathcal{Y}^{i-1}\times\mathcal{Z}^{i}\to\mathcal{V}$). The capacity of the AVCF is the supremum of $R\ge0$ such that for any $n_{0}\in\mathbb{Z}_{>0}$, $\epsilon>0$, there exists $n\ge n_{0}$ and a coding scheme with $L=\lfloor2^{nR}\rfloor$ and maximum error probability $\epsilon$. We now show that the AVCF has the same capacity and finite-blocklength performance as the corresponding game-theoretic channel.
\end{defn}
\smallskip{}

\begin{thm}
\label{thm:avcf}Given an AVCF $p_{Y|X,Z,V}$ with finite $\mathcal{X},\mathcal{Z},\mathcal{V},\mathcal{Y}$, construct a game-theoretic channel $W:\mathcal{X}\to\mathrm{DCCone}(\mathcal{Y})$, 
\[
W(x):=\bigcup_{z\in\mathcal{Z}}\bigcap_{v\in\mathcal{V}}(p_{Y|X,Z,V}(\cdot|x,z,v))^{\circ}.
\]
The channel coding game $(W,n,L,\epsilon)$ (Definition \ref{def:game_channel} or \ref{def:game_channel_one}) is a guaranteed win if and only if the AVCF has a coding scheme with maximal error probability $\epsilon$. This continues to hold if we require $\sum_{j=1}^{i}w_{j}(y_{j})\ge-\epsilon$ for all $i$ in Definition \ref{def:game_channel}. Hence, $\mathrm{C}(W)$ equals the capacity of the AVCF.
\end{thm}
\smallskip{}

Theorem \ref{thm:avcf} reveals an interesting duality between the causal input $Z_{i}$ and the adversarial input $V_{i}$. Consider a channel $p_{Y|X,S}$ with state $S\in\mathcal{S}$. If the encoder controls $S$ causally (i.e., we consider the AVCF $p_{Y|X,Z,V}$ with $Z=S$, and $V$ is a constant), then the game-theoretic channel is $W_{\mathrm{E}}(x):=\bigcup_{s\in\mathcal{S}}(p_{Y|X,S}(\cdot|x,s))^{\circ}$. If the adversary controls $S$ causally (i.e., $V=S$, and $Z$ is a constant), then the channel is $W_{\mathrm{A}}(x):=\bigcap_{s\in\mathcal{S}}(p_{Y|X,S}(\cdot|x,s))^{\circ}$. The two DC cones $W_{\mathrm{E}}(x),W_{\mathrm{A}}(x)$ are duals of each other (Definition \ref{def:dual}). Taking the dual swaps the roles of the encoder and the adversary.

Theorem \ref{thm:cap} gives an upper bound $\mathrm{C}(W)\le\sup_{p_{X,Z}}\inf_{p_{V|X,Z}}I(X,Z;Y)$, which can also be seen by observing that if the adversary generates $V_{i}|(X_{i},Z_{i})\sim p_{V|X,Z}$, then the capacity is upper-bounded by the feedback capacity of the induced DMC $p_{Y|X,Z}$.

We now prove Theorem \ref{thm:avcf}.

\begin{IEEEproof}
We first prove the ``if'' direction, and assume there exists a code $(x^{n},(z_{i})_{i\in[n]},\hat{m})$ for the AVCF with maximal error probability at most $\epsilon$. We now describe a strategy for the encoder in the channel coding game. For $m\in[L]$, $y^{i}\in\mathcal{Y}^{i}$, $0\le i\le n$, let 
\[
P_{e}(m,y^{i}):=\max_{(v_{j})_{j>i}}\mathbb{P}(\hat{m}(y^{i},Y_{i+1}^{n})\neq m\,|\,M=m,Y^{i}=y^{i})
\]
be the worst-case error probability conditional on $M=m$, $X^{n}=x^{n}=x^{n}(m)$ and $Y^{i}=y^{i}$, maximized over the adversary's strategies $(v_{j})_{j>i}$ after time $i$. Since the maximum error probability is $\epsilon$, we have $P_{e}(m,\emptyset)\le\epsilon$, where $\emptyset$ denotes the sequence with length $0$, and the recursive formula
\begin{align}
 & P_{e}(m,y^{i-1})\nonumber \\
 & =\max_{v\in\mathcal{V}}\sum_{y}p_{Y|X,Z,V}(y|x_{i},z_{i}(m,y^{i-1}),v)P_{e}(m,y^{i-1},y)\nonumber \\
 & =\max_{v\in\mathcal{V}}\langle p_{i,v},P_{e}(m,y^{i-1},\cdot)\rangle,\label{eq:pf_max_pe}
\end{align}
where $p_{i,v}:=(y\mapsto p_{Y|X,Z,V}(y|x_{i},z_{i}(m,y^{i-1}),v))\in\mathbb{R}^{\mathcal{Y}}$, $P_{e}(m,y^{i-1},y)$ is $P_{e}$ applied to $m$ and $(y^{i-1},y)\in\mathcal{Y}^{i}$, and  $P_{e}(m,y^{i-1},\cdot)=(y\mapsto P_{e}(m,y^{i-1},y))\in\mathbb{R}^{\mathcal{Y}}$. At time $i$, the encoder chooses
\[
w_{i}(y')=P_{e}(m,y^{i-1},y')-P_{e}(m,y^{i-1}).
\]
To check that $w_{i}\in W(x_{i})$, letting $z_{i}=z_{i}(m,y^{i-1})$, for any $v\in\mathcal{V}$, we have 
\begin{align*}
 & \langle p_{Y|X,Z,V}(\cdot|x_{i},z_{i},v),w_{i}\rangle\\
 & =\langle p_{i,v},P_{e}(m,y^{i-1},\cdot)-P_{e}(m,y^{i-1})\rangle\,\le\,0
\end{align*}
due to (\ref{eq:pf_max_pe}). Hence, $w_{i}\in W(x_{i})$. The decoder outputs $\hat{m}=\hat{m}(y^{n})$. We have 
\begin{align*}
 & \sum_{i=1}^{n}w_{i}(y_{i})-\mathbf{1}\{\hat{m}\neq m\}\\
 & =P_{e}(m,y^{n})-P_{e}(m,\emptyset)-\mathbf{1}\{\hat{m}\neq m\}\,\ge\,-\epsilon.
\end{align*}
Therefore, the team wins. Also note that $\sum_{j=1}^{i}w_{j}(y_{j})=P_{e}(m,y^{i})-P_{e}(m,\emptyset)\ge-\epsilon$, so the team is never in a deficit more than $\epsilon$.

We then prove the ``only if'' direction, and consider a winning strategy for the channel coding game. We now construct a coding scheme. Let $x^{n}(m)$ be the $x^{n}$ produced by the encoder upon message $m$, and $z_{i}(m,y^{i-1})$ be a $z_{i}\in\mathcal{Z}$ such that the $w_{i}$ produced by the encoder upon $m,y^{i-1}$ satisfies $w_{i}\in\{p_{Y|X,Z,V}(\cdot|x_{i},z_{i},v):\,v\in\mathcal{V}\}^{\circ}$ (such a $z_{i}$ exists since $w_{i}\in W(x_{i})$). Let $\hat{m}(y^{n})$ be the $\hat{m}$ produced by the decoder upon $y^{n}$. We now show that for every $m$ and every adversary's strategy $(v_{i})_{i}$ for the AVCF, the error probability is at most $\epsilon$. Fix $m$ and $(v_{i})_{i}$. For the channel coding game, assume the adversary produces $y_{i}$ following $p_{Y|X,Z,V}(\cdot|x_{i},z_{i},v_{i}(m,y^{i-1},z^{i}))$. Then $x^{n},y^{n},z^{n}$ will have the same distribution as in the AVCF when the adversary uses the strategy $(v_{i})_{i}$. We have $\mathbb{E}[w_{i}(y_{i})]\le0$ since $w_{i}\in\{p_{Y|X,Z,V}(\cdot|x_{i},z_{i},v):\,v\in\mathcal{V}\}^{\circ}$. Since the team wins, we have
\begin{align*}
-\epsilon & \le\mathbb{E}\Big[\sum_{i=1}^{n}w_{i}(y_{i})-\mathbf{1}\{\hat{m}\neq m\}\Big]\le-\mathbb{P}(\hat{m}\neq m).
\end{align*}
Hence, the error probability is at most $\epsilon$.

\end{IEEEproof}
\smallskip{}

We now give an achievability proof of Theorem \ref{thm:cap} using Theorem \ref{thm:avcf}.
\begin{IEEEproof}
[Achievability of Theorem \ref{thm:cap}]Consider any $W:\mathcal{X}\to\mathrm{DCCone}(\mathcal{Y})$ (with finite $\mathcal{X},\mathcal{Y}$) that is pointwise closed and concave, i.e., $\mathbb{R}^{\mathcal{Y}}\backslash W(x)$ is open and convex for all $x$. If there exists $x$ such that $W(x)=\mathbb{R}^{\mathcal{Y}}$, then the encoder can use $x$ to perform arbitrage, and an arbitrarily large transmission rate is possible. Hence, we assume $W(x)\neq\mathbb{R}^{\mathcal{Y}}$ for all $x$. We first consider the case where $\mathbb{R}^{\mathcal{Y}}\backslash W(x)$ is a polyhedral cone for all $x$, i.e., we can find a conditional distribution $p_{Y|X,Z}$ from $\mathcal{X}\times\mathcal{Z}$ to $\mathcal{Y}$ (where $\mathcal{Z}$ is finite) such that $\mathbb{R}^{\mathcal{Y}}\backslash W(x)=\{a\in\mathbb{R}^{\mathcal{Y}}:\,\forall z.\,\langle p_{Y|X,Z}(\cdot|x,z),a\rangle>0\}$. Then $W(x)=\bigcup_{z}(p_{Y|X,Z}(\cdot|x,z))^{\circ}$ corresponds to an AVCF without adversary (i.e., $\mathcal{V}$ is a singleton set), that is a DMC with noncausal input $X$ and causal input $Z$. For the achievability, we can treat $Z$ as a noncausal input and consider any scheme for the DMC from $(X,Z)$ to $Y$ without feedback, which can achieve the capacity $\max_{p_{X,Z}}I(X,Z;Y)$. By Theorem \ref{thm:avcf}, we have $\mathrm{C}(W)\ge\max_{p_{X,Z}}I(X,Z;Y)$. Since $W(\mathcal{X})=\bigcup_{x,z}(p_{Y|X,Z}(\cdot|x,z))^{\circ}$, by Proposition \ref{prop:i_prop}, we have $\mathrm{I}(W(\mathcal{X}))=\max_{p_{X,Z}}I(X,Z;Y)\le\mathrm{C}(W)$.

We then consider the case where $\mathbb{R}^{\mathcal{Y}}\backslash W(x)$ is not a polyhedral cone. Using the alternative characterization in Proposition \ref{prop:i_prop}, we have $\mathrm{I}(W(\mathcal{X}))=\sup_{p_{T}}\inf_{p_{Y|T}}I(T;Y)$, where $p_{T}$ ranges over finite discrete distributions over $W(\mathcal{X})$, and $p_{Y|T}$ ranges over conditional distributions from $W(\mathcal{X})$ to $\mathcal{Y}$ satisfying that $\langle t,p_{Y|T}(\cdot|t)\rangle\le0$ for all $t\in W(\mathcal{X})$. Hence, it suffices to show that for every $p_{T}$, we can achieve the rate $\inf_{p_{Y|T}}I(T;Y)$. 

Fix any $p_{T}$, and assume its support is $S\subseteq W(\mathcal{X})$ which is finite. For $a\in S$, let $x(a)\in\mathcal{X}$ such that $a\in W(x(a))$, and let $r(a)\in\mathbb{R}^{\mathcal{Y}}$ be a normal of a separating hyperplane between $\{\gamma a:\,\gamma\ge0\}$ (closed and convex) and $\mathbb{R}^{\mathcal{Y}}\backslash W(x(a))$ (open and convex), i.e., $\langle r(a),a\rangle\le0$ and $\langle r(a),b\rangle>0$ for all $b\in\mathbb{R}^{\mathcal{Y}}\backslash W(x(a))$. In particular, since $W(x(a))\neq\mathbb{R}^{\mathcal{Y}}$, we have $\langle r(a),b\rangle>0$ for all $b\in\mathbb{R}_{>0}^{\mathcal{Y}}$, so $r(a)\in\mathbb{R}_{>0}^{\mathcal{Y}}\backslash\{0\}$. Hence, we can assume $r(a)\in\Delta_{\mathcal{Y}}$. Define $\tilde{\mathcal{X}}:=\bigcup_{a\in S}\{x(a)\}\subseteq\mathcal{X}$, and $\tilde{W}:\tilde{\mathcal{X}}\to\mathrm{DCCone}(\mathcal{Y})$, $\tilde{W}(x):=\bigcup_{a\in S:\,x(a)=x}r(a)^{\circ}$. We have $a\in\tilde{W}(x(a))$ for $a\in S$ (so $S\subseteq\tilde{W}(\tilde{\mathcal{X}})$), and $\tilde{W}(x)\subseteq W(x)$ for all $x\in\tilde{\mathcal{X}}$ since $r(a)^{\circ}$ and $\mathbb{R}^{\mathcal{Y}}\backslash W(x(a))$ are disjoint. Hence, any strategy for the channel coding game for $\tilde{W}$ is also valid for $W$. Also, $\tilde{W}$ is pointwise closed and concave, and $\mathbb{R}^{\mathcal{Y}}\backslash\tilde{W}(x)$ is a polyhedral cone for all $x\in\tilde{\mathcal{X}}$. Hence, we have $\mathrm{C}(W)\ge\mathrm{C}(\tilde{W})\ge\mathrm{I}(\tilde{W}(\tilde{\mathcal{X}}))$. Since the support of $p_{T}$ is $S\subseteq\tilde{W}(\tilde{\mathcal{X}})$, by the alternative characterization in Proposition \ref{prop:i_prop}, $\mathrm{I}(\tilde{W}(\tilde{\mathcal{X}}))\ge\inf_{p_{Y|T}}I(T;Y)$, which completes the proof.

\end{IEEEproof}
\smallskip{}

\subsection{Hamming Adversarial Channel}

Another widely studied adversarial channel is the Hamming adversarial channel \cite{hamming1950error}, where the encoder sends $x^{n}\in\mathcal{X}^{n}$, and the adversary can choose any $y^{n}\in\mathcal{X}^{n}$ with Hamming distance $|\{i:x_{i}\neq y_{i}\}|\le\zeta n$. The goal is to send $m\in[L]$ with zero error. Assume that the encoder can purchase an insurance at $\$\gamma\ge0$ for each $i=1,\ldots,n$, which pays $\$\gamma/p$ if $y_{i}\neq x_{i}$, where $p=\zeta\epsilon$. This is possible for a game-theoretic channel $W_{\mathrm{H}}:\mathcal{X}\to\mathrm{DCCone}(\mathcal{X})$, 
\[
W_{\mathrm{H}}(x):=\big\{ a\in\mathbb{R}^{\mathcal{X}}:a\le\gamma((\mathbf{1}-\mathbf{e}_{x})/p-\mathbf{1})\;\text{for some}\;\gamma\ge0\big\},
\]
where $\mathbf{e}_{x},\mathbf{1}\in\mathbb{R}^{\mathcal{X}}$ are the $x$-th standard basis vector and the all-one vector, respectively. The encoder can choose $\gamma=\epsilon/n$ for each $i$ to guarantee that either $|\{i:x_{i}\neq y_{i}\}|\le\zeta n$ and the maximum loss is at most $\$\epsilon$, or $|\{i:x_{i}\neq y_{i}\}|>\zeta n$ and the gain from the insurance is at least $\$(1-\epsilon)$.  This immediately implies the following sufficient condition for coding with the game-theoretic channel $W_{\mathrm{H}}$.

\smallskip{}

\begin{prop}
\label{prop:Hamming}For finite $\mathcal{X},\mathcal{Y}$, the channel coding game $(W_{\mathrm{H}},n,L,\epsilon)$ (Definition \ref{def:game_channel} or \ref{def:game_channel_one}) with $p=\zeta\epsilon$ is a guaranteed win if a zero-error coding scheme for the Hamming adversarial channel exists with blocklength $n$ and message cardinality $L$.
\end{prop}
\smallskip{}

Nevertheless, this is not a necessary condition since the encoder can choose a different ``weight'' $\gamma$ for different coordinate $i$ in game-theoretic channel coding. Hence, the game-theoretic channel does not exactly generalizes the Hamming adversarial channel (unless we generalize the game-theoretic channel to allow imposing that the $\gamma$ must be constant for each $i$).

\smallskip{}

\section{Single-Cone Channels}\label{sec:single_cone}

This subsection focuses on channels where $\mathcal{X}=\{0\}$ is a singleton set, which are channels with noiseless feedback such as zero-error channels with feedback (Section \ref{subsec:adv_feedback}), DMCs with feedback (Section \ref{subsec:dmc_feedback}), and AVCFs with $\mathcal{X}=\{0\}$ (Section \ref{subsec:avcf}). We call them \emph{single-cone channels} since they are characterized by one pricing DC cone $A=W(0)\in\mathrm{DCCone}(\mathcal{Y})$. Single-cone channels have a rich structure, and support various operations (let $A\in\mathrm{DCCone}(\mathcal{Y})$, $B\in\mathrm{DCCone}(\mathcal{Z})$):
\begin{itemize}
\item   The \emph{union channel} of $A,B$ is $A\cup B$ (assume $\mathcal{Y}=\mathcal{Z}$). This means the encoder can choose whether to use channel $A$ or $B$ at each time, though the output set is still $\mathcal{Y}$ so the decoder does not know which channel is chosen.
\item The \emph{intersection channel} is $A\cap B$ (assume $\mathcal{Y}=\mathcal{Z}$). This means the encoder chooses an input for both channels $A,B$ at each time, and the adversary chooses one of the channels and use that channel. 
\item The \emph{Minkowski sum channel} is $A+B$ (assume $\mathcal{Y}=\mathcal{Z}$). This means the encoder chooses an input for both channels $A,B$, and the adversary needs to satisfy the constraints of both inputs. This is not particularly interesting as it usually leads to arbitrage.
\item The \emph{dual channel} of $A$ is $A^{\diamondsuit}$, corresponding to swapping the roles of the encoder and the adversary (see Section \ref{subsec:adv_feedback}).
\item The \emph{implication channel} is (assume $\mathcal{Y}=\mathcal{Z}$)
\[
A\to B:=(B\backslash A)+\mathbb{R}_{\le0}^{\mathcal{Y}}.
\]
Note that $A\to B$ is the smallest pricing DC cone such that $(A\to B)\cup A\supseteq B$.  This turns $\mathrm{DCCone}(\mathcal{Y})\backslash\{\emptyset\}$ into a Heyting algebra \cite{heyting1930formalen} in a similar manner as \cite{li2025coding}, with bottom element $\mathbb{R}^{\mathcal{Y}}$, top element $\mathbb{R}_{\le0}^{\mathcal{Y}}$, ordering $\supseteq$, meet $\cup$, join $\cap$ and implication $\to$. 
\item The \emph{convex combination channel} is $A\stackrel{\lambda}{\oplus}B$ where $0\le\lambda\le1$ (assume $\mathcal{Y}=\mathcal{Z}$). This means the encoder chooses an input for both channels $A,B$, and then channel $A$ is used with probability $1-\lambda$, or channel $B$ is used with probability $\lambda$.
\item In particular, $A\stackrel{\epsilon}{\oplus}\mathbb{R}_{\le0}^{\mathcal{Y}}$ is a ``robust'' version of channel $A$, where we allow a probability $\epsilon$ for the output to be arbitrary. A coding scheme for this channel will also work for every channel that is close to $A$. For example, if $A=\bigcup_{x\in\mathcal{X}}(p_{Y|X}(\cdot|x))^{\circ}$ is a DMC, and $B=\bigcup_{x\in\mathcal{X}}(q_{Y|X}(\cdot|x))^{\circ}$ is another DMC that is close to $p_{Y|X}$ in terms of $\infty$-R\'{e}nyi divergence \cite{renyi1961measures}, that is, $D_{\infty}(p_{Y|X}(\cdot|x)\Vert q_{Y|X}(\cdot|x))\le-\log(1-\epsilon)$ for every $x$, or equivalently,
\[
\min_{x,y}q_{Y|X}(y|x)/p_{Y|X}(y|x)\ge1-\epsilon,
\]
then $A\stackrel{\epsilon}{\oplus}\mathbb{R}_{\le0}^{\mathcal{Y}}\subseteq B$, so a scheme for $A\stackrel{\epsilon}{\oplus}\mathbb{R}_{\le0}^{\mathcal{Y}}$ will also work for $B$.\footnote{To show this, let $r(y|x)=\epsilon^{-1}(q(y|x)-(1-\epsilon)p(y|x))$ (we omit the subscripts ``$Y|X$''). We have $q(y|x)=(1-\epsilon)p(y|x)+\epsilon r(y|x)$, and hence $q(\cdot|x)^{\circ}=p(\cdot|x)\stackrel{\epsilon}{\oplus}r(\cdot|x)\supseteq p(\cdot|x)\stackrel{\epsilon}{\oplus}\mathbb{R}_{\le0}^{\mathcal{Y}}$. Taking union over $x$ gives $B\supseteq A\stackrel{\epsilon}{\oplus}\mathbb{R}_{\le0}^{\mathcal{Y}}$. } Note that the DC cone corresponding to the decoding constraint in (\ref{eq:channel_require}) is the robust version of the noiseless channel.
\item The \emph{causal product channel} is $A\ltimes B$. To use $A\ltimes B$ once means that the encoder first uses the channel $A$, observes the channel output via feedback, and then uses the channel $B$. The decoder observes both channel outputs.
\item The \emph{adversarial product channel} is $A\otimes B$. This means the encoder chooses an input for both channels $A,B$ simultaneously, and the adversary produces the outputs of the two channels such that the constraints of both $A$ and $B$ are satisfied, but the outputs can be coupled in an adversarial manner. 
\item The \emph{sum channel} is $A\sqcup B\in\mathrm{DCCone}(\mathcal{Y}\cup\mathcal{Z})$ (assuming $\mathcal{Y},\mathcal{Z}$ are disjoint),\footnote{If $\mathcal{Y},\mathcal{Z}$ are not disjoint, we consider their disjoint union $(\{1\}\times\mathcal{Y})\cup(\{2\}\times\mathcal{Z})$ instead of $\mathcal{Y}\cup\mathcal{Z}$.}
\begin{align*}
A\sqcup B & :=\big\{ c\in\mathbb{R}^{\mathcal{Y}\cup\mathcal{Z}}:\,(\exists a\in A.\,\forall y\in\mathcal{Y}.\,c(y)\le a(y))\\
 & \quad\quad\mathrm{or}\;\;(\exists b\in B.\,\forall z\in\mathcal{Z}.\,c(z)\le b(z))\big\}.
\end{align*}
The encoder can either choose a portfolio $a\in A$ for outcomes in $\mathcal{Y}$ and set arbitrarily large payoffs for outcomes in $\mathcal{Z}$ (which forces the adversary to pick an outcome in $\mathcal{Y}$), or choose a portfolio $b\in B$ for $\mathcal{Z}$ and set arbitrarily large payoffs for $\mathcal{Y}$ (which forces the adversary to pick an outcome in $\mathcal{Z}$). Hence, the encoder can choose whether to use $A$ or $B$, and the decoder knows the choice (unlike the union channel $A\cup B$ where the decoder does not know). This corresponds to the conventional notion of sum channel \cite{shannon1948mathematical}.
\end{itemize}
\smallskip{}

\section{Duality and Adversarial Cost Channel Coding}\label{sec:adv_cost}

In the channel coding game (Definition \ref{def:game_channel}), the encoder chooses the portfolios, and the adversary chooses the channel outputs. We can consider a ``dual channel coding game'' where their roles are swapped, so the adversary chooses the portfolios, and the encoder chooses the channel outputs. If the adversary chooses portfolio $t$, and the encoder chooses $y$, then the encoder loses $t(y)$. Hence, the portfolio can be regarded as an adversarial cost function. The game is described as follows.

\smallskip{}

\begin{defn}
\label{def:game_channel_adv} The \emph{adversarial cost channel coding game} $(T,n,L,\epsilon)$ with cost channel $T:\mathcal{X}\to\mathrm{DCCone}(\mathcal{Y})$, blocklength $n$, message cardinality $L\in\mathbb{Z}_{>0}$ and maximum loss $0<\epsilon<1$ is defined as follows:
\begin{itemize}
\item Adversary produces $m\in[L]$, observed by encoder.
\item Encoder produces $x^{n}\in\mathcal{X}^{n}$, observed by adversary.
\item For $i=1,\ldots,n$:
\begin{itemize}
\item Adversary produces $t_{i}\in T(x_{i})$, observed by encoder.
\item Encoder produces $y_{i}\in\mathcal{Y}$, observed by adversary.
\end{itemize}
\item Decoder observes $y^{n}$, produces $\hat{m}\in[L]$.
\item The team wins if $-\sum_{i=1}^{n}t_{i}(y_{i})-\mathbf{1}\{\hat{m}\neq m\}\ge-\epsilon$.
\end{itemize}
\end{defn}
\smallskip{}

We show that the adversarial cost channel coding game with a cost channel $T$ is equivalent to the channel coding game (Definition \ref{def:game_channel}) with the ``dual channel'' $W$ where $W(x)=(T(x))^{\diamondsuit}$ is the dual DC cone of $T(x)$ (Definition \ref{def:dual}). Hence, all results on the channel coding game also apply to the adversarial cost channel coding game after taking the dual.

\smallskip{}

\begin{thm}
Assume $\mathcal{Y}$ is finite, and $T(x)\notin\{\emptyset,\mathbb{R}^{\mathcal{Y}}\}$ for all $x$. The adversarial cost channel coding game $(T,n,L,\epsilon)$ is a guaranteed win if and only if the channel coding game $(W,n,L,\epsilon)$ with $W(x)=(T(x))^{\diamondsuit}$ is a guaranteed win.
\end{thm}
\smallskip{}

\begin{IEEEproof}
Consider a winning strategy for the channel coding game (CCG). Assume that the encoder produces $w_{i}=w_{i}(m,y^{i-1})\in W(x_{i})$ upon observing $m,y^{i-1}$ at time $i$. We now design a strategy for the adversarial cost channel coding game (ACCCG), which is the same as that for CCG except that the encoder produces $y_{i}$ which satisfies $t_{i}(y_{i})+w_{i}(m,y^{i-1})(y_{i})\le0$ at time $i$ (this is possible since $w_{i}(m,y^{i-1})\in(T(x_{i}))^{\diamondsuit}$). Then $-\sum_{i}t_{i}(y_{i})-\mathbf{1}\{\hat{m}\neq m\}\ge\sum_{i}w_{i}(m,y^{i-1})(y_{i})-\mathbf{1}\{\hat{m}\neq m\}\ge-\epsilon$. Hence, this is a winning strategy for ACCCG.

For the other direction, consider a winning strategy for ACCCG. Assume that the encoder produces $y_{i}=y_{i}(m,y^{i-1},t_{i})$ at time $i$ (we can assume that $y_{i}$ does not depend on $t^{i-1}$ since $t^{i-1}$ only influences the game through $y^{i-1}$; the state of the game at time $i$ is $(m,y^{i-1})$). We now modify this strategy to give a strategy for CCG, where the encoder would produce $w_{i}$ with $w_{i}(y')=-\sup_{t\in T(x_{i}):\,y_{i}(m,y^{i-1},t)=y'}t(y')$ upon observing $m,y^{i-1}$ at time $i$. Note that $w_{i}(y')>-\infty$, or else the adversary in ACCCG can choose $t\in T(x_{i})$ with $y_{i}(m,y^{i-1},t)=y'$ and $t(y')$ being arbitrarily large, and choose $0\in T(x_{i'})$ for subsequent turns $i'>i$ to make the payoff of the team arbitrarily small, so the team cannot win. If there is no $t\in T(x_{i})$ with $y_{i}(m,y^{i-1},t)=y'$, take $w_{i}(y')$ large enough so that $\sum_{i'=1}^{i-1}w_{i'}(y_{i'})+w_{i}(y')\ge1$, so in case if the adversary in CCG chooses $y_{i}=y'$, the encoder can simply choose $0\in w_{i'}(x_{i'})$ for subsequent turns $i'>i$ to win the game. To check that $w_{i}\in W(x_{i})=(T(x_{i}))^{\diamondsuit}$, for any $t\in T(x_{i})$, letting $y'=y_{i}(m,y^{i-1},t)$, we have $w_{i}(y')+t(y')\le0$, so $w_{i}\in W(x_{i})$. The total payoff is 
\begin{align*}
 & \sum_{i}w_{i}(y_{i})-\mathbf{1}\{\hat{m}\neq m\}\\
 & =\inf_{t_{i}\in T(x_{i}):\,y_{i}(m,y^{i-1},t)=y_{i},\,i\in[n]}\sum_{i}(-t_{i}(y_{i}))-\mathbf{1}\{\hat{m}\neq m\}\\
 & \ge-\epsilon,
\end{align*}
since $y_{i}=y_{i}(m,y^{i-1},t_{i})$ is a winning strategy for ACCCG which guarantees a payoff at least $-\epsilon$ regardless of $t_{i}$. Hence, this is a winning strategy for CCG.
\end{IEEEproof}
\smallskip{}

\section{Lossless Source Coding Game}\label{sec:lossless}

We have seen that a random channel can be simulated by a deterministic game. Using a similar construction, we can simulate a random source, and define a game corresponding to lossless source coding.

\smallskip{}

\begin{defn}
\label{def:game_lossless} The \emph{lossless source coding game }is parametrized by a tuple $(A,n,L,\epsilon)$, where $A\in\mathrm{DCCone}(\mathcal{X})$, $n,L\in\mathbb{Z}_{>0}$, and $0<\epsilon<1$. It is defined as follows:
\begin{itemize}
\item For $i=1,\ldots,n$:
\begin{itemize}
\item Encoder produces $w_{i}\in A$, observed by the adversary.
\item Adversary produces $x_{i}\in\mathcal{X}$, observed by the encoder.
\end{itemize}
\item Encoder produces $m\in[L]$.
\item Decoder observes $m$, produces $\hat{x}^{n}\in\mathcal{X}^{n}$.
\item The team (encoder and decoder) wins if $\sum_{i=1}^{n}w_{i}(x_{i})-\mathbf{1}\{\hat{x}^{n}\neq x^{n}\}\ge-\epsilon$.
\end{itemize}
The \emph{entropy} of $A$, denoted as $\mathrm{H}(A)$, is the infimum of the set of $R\ge0$ such that for every $n_{0}\in\mathbb{Z}_{>0}$, $0<\epsilon<1$, there exists $n\ge n_{0}$ such that the lossless source coding game $(A,n,\lfloor2^{nR}\rfloor,\epsilon)$ is a guaranteed win. If $A=\mathbb{R}^{\mathcal{X}}$, take $\mathrm{H}(A)=-\infty$.
\end{defn}
\smallskip{}

Similar to Definition \ref{def:game_channel_one}, we can replace the loop with 1) encoder produces $w\in A^{\ltimes n}$; and 2) adversary observes $w$ and produces $x^{n}\in\mathcal{X}^{n}$. The team wins if $w(x^{n})-\mathbf{1}\{\hat{x}^{n}\neq x^{n}\}\ge-\epsilon$. 

The lossless source coding game generalizes the conventional notion of source coding, as shown below.

\smallskip{}

\begin{thm}
\label{thm:lossless}Let $p_{X}\in\Delta_{\mathcal{X}}$ for a finite $\mathcal{X}$, and $A=p_{X}^{\circ}$. The lossless source coding game $(A,n,L,\epsilon)$ is a guaranteed win if and only if there exists a lossless source coding scheme $(f_{n},g_{n})$ for $p_{X}$ with error probability at most $\epsilon$, i.e., $f_{n}:\mathcal{X}^{n}\to[L]$ and $g_{n}:[L]\to\mathcal{X}^{n}$ with $\mathbb{P}(g_{n}(f_{n}(X^{n}))\neq X^{n})\le\epsilon$ where $X_{i}\stackrel{\mathrm{iid}}{\sim}p_{X}$. 
\end{thm}
\smallskip{}

\begin{IEEEproof}
We first prove the ``if'' direction, and consider a coding scheme $(f_{n},g_{n})$ with $\mathbb{P}(g_{n}(f_{n}(X^{n}))\neq X^{n})\le\epsilon$ where $X_{i}\stackrel{\mathrm{iid}}{\sim}p_{X}$. Let $P_{e}(x^{i}):=\mathbb{P}(g_{n}(f_{n}(X^{n}))\neq X^{n}\,|\,X^{i}=x^{i})$ for $i\in\{0,\ldots,n\}$, $x^{i}\in\mathcal{X}^{i}$. Then $P_{e}(\emptyset)\le\epsilon$. To construct a strategy for the game, at time $i$, the encoder produces $w_{i}(x')=P_{e}((x^{i-1},x'))-P_{e}(x^{i-1})$. Since $P_{e}(x^{i-1})=\sum_{x'}p_{X}(x')P_{e}((x^{i-1},x'))$, we have $\langle p_{X},w_{i}\rangle=0$, so $w_{i}\in A$. The total payoff is $\sum_{i=1}^{n}w_{i}(x_{i})-\mathbf{1}\{\hat{x}^{n}\neq x^{n}\}=P_{e}(x^{n})-P_{e}(\emptyset)-\mathbf{1}\{\hat{x}^{n}\neq x^{n}\}\ge-\epsilon$. Hence, this is a winning strategy.

We then prove the ``only if'' direction, and consider a winning strategy for the game. Assume the adversary produces $x_{i}$ i.i.d. following $p_{X}$. Then $\mathbb{E}[w_{i}(x_{i})]=\langle p_{X},w_{i}\rangle\le0$. Since the team wins, we have $-\epsilon\le\mathbb{E}[\sum_{i=1}^{n}w_{i}(x_{i})-\mathbf{1}\{\hat{x}^{n}\neq x^{n}\}]\le-\mathbb{P}(\hat{x}^{n}\neq x^{n})$, so $\mathbb{P}(\hat{x}^{n}\neq x^{n})\le\epsilon$. Hence, the encoder's strategy for mapping $x^{n}$ to $m$, and the decoder's strategy for mapping $m$ to $\hat{x}^{n}$, is a valid coding scheme for lossless source coding.
\end{IEEEproof}
\smallskip{}

The optimal rate of the game can be exactly characterized, as shown in the following result.

\smallskip{}

\begin{thm}
\label{thm:entropy}For finite $\mathcal{X}$, the entropy of $A$ (optimal rate of the lossless source coding game) can be given as
\[
\mathrm{H}(A)=\inf_{a\in A}\,\sup_{p\in\Delta_{\mathcal{X}}:\,\langle p,a\rangle\le0}H(p),
\]
where $H(p)$ denotes the entropy of the distribution $p$. We have $\mathrm{H}(A)=-\infty$ if $A=\mathbb{R}^{\mathcal{X}}$, and $\mathrm{H}(A)=\infty$ if $A=\emptyset$.
\end{thm}
\smallskip{}

\begin{IEEEproof}
 We focus on the case $A\neq\mathbb{R}^{\mathcal{X}}$, $A\neq\emptyset$ (otherwise the result follows from the definition). We first prove that we can achieve any rate greater than $R^{*}:=\inf_{a\in A}\sup_{p:\,\langle p,a\rangle\le0}H(p)$. Consider any $a\in A$. The strategy of the encoder is to output $w_{i}=(\gamma/n)a$ for all $i$, where $\gamma>0$ will be chosen later. The team wins as long as $\sum_{i=1}^{n}w_{i}(x_{i})=(\gamma/n)\sum_{i}a(x_{i})\ge1-\epsilon$, regardless of whether $\hat{x}^{n}\neq x^{n}$. Hence, we only have to design a code for transmitting $x^{n}$ in $S:=\{x^{n}\in\mathcal{X}^{n}:\,n^{-1}\sum_{i}a(x_{i})<(1-\epsilon)/\gamma\}$, which is possible as long as $|S|\le L=\lfloor2^{nR}\rfloor$. Applying Sanov' theorem \cite{sanov1957probability,dembo2009large} on the uniform distribution $u\in\Delta_{\mathcal{X}}$ over $\mathcal{X}$, we have
\[
\frac{|S|}{|\mathcal{X}|^{n}}\le(n+1)^{|\mathcal{X}|}2^{-n\inf_{p:\,\langle p,a\rangle<(1-\epsilon)/\gamma}D(p\Vert u)}.
\]
Since $D(p\Vert u)=\log|\mathcal{X}|-H(p)$, we have $|S|\le(n+1)^{|\mathcal{X}|}2^{n\sup_{p:\,\langle p,a\rangle<(1-\epsilon)/\gamma}H(p)}$. Hence, we can achieve any rate $R>\sup_{p:\,\langle p,a\rangle<(1-\epsilon)/\gamma}H(p)$. Taking $\gamma\to\infty$, we can achieve any rate $R>\sup_{p:\,\langle p,a\rangle\le0}H(p)$. 

We then show the converse that any achievable rate must satisfy $R\ge R^{*}=\inf_{a\in A}\sup_{p:\,\langle p,a\rangle\le0}H(p)$. Consider a winning strategy. Assume that the adversary generates $x_{i}$ according to the $p$ that attains $\sup_{p:\,\langle p,w_{i}\rangle\le0}H(p)$. We have $\mathbb{E}[w_{i}(x_{i})]\le0$. Since the team wins, we have $-\epsilon\le\mathbb{E}[\sum_{i=1}^{n}w_{i}(x_{i})-\mathbf{1}\{\hat{x}^{n}\neq x^{n}\}]\le-\mathbb{P}(\hat{x}^{n}\neq x^{n})$, so $\mathbb{P}(\hat{x}^{n}\neq x^{n})\le\epsilon$. Also, we have 
\begin{align*}
H(x^{n}) & =\sum_{i=1}^{n}H(x_{i}|x^{i-1})\\
 & =\sum_{i=1}^{n}\sup_{p:\,\langle p,w_{i}\rangle\le0}H(p)\;\ge\;nR^{*}.
\end{align*}
By Fano's inequality \cite{fano1961transmission}, $H(x^{n})\le H(\hat{x}^{n})+H(x^{n}|\hat{x}^{n})\le(1+\epsilon)\log L+1\le(1+\epsilon)nR+1$. Taking $\epsilon\to0$, $n\to\infty$ gives $R\ge R^{*}$.
\end{IEEEproof}
\smallskip{}

The information capacity and the entropy of a DC cone can be expressed in terms of the following more general quantity.

\smallskip{}

\begin{defn}
For $A,B\in\mathrm{DCCone}(\mathcal{Y})$, their \emph{relative entropy} is
\[
\mathrm{D}(A\Vert B):=\sup_{b\in B}\,\inf_{q\in b^{\circ\Delta}}\,\sup_{a\in A}\,\inf_{p\in a^{\circ\Delta}}D(p\Vert q),
\]
where $a^{\circ\Delta}:=\{p\in\Delta_{\mathcal{Y}}:\,\langle p,a\rangle\le0\}$.
\end{defn}
\smallskip{}

We can check that for finite $\mathcal{Y}$ and $A\in\mathrm{DCCone}(\mathcal{Y})\backslash\{\emptyset,\mathbb{R}^{\mathcal{Y}}\}$, we have $\mathrm{I}(A)=\mathrm{D}(A\Vert\mathbb{R}_{\le0}^{\mathcal{Y}})$ and $\mathrm{H}(A)=\log|\mathcal{Y}|-\mathrm{D}(A\Vert(\mathrm{Unif}(\mathcal{Y}))^{\circ})$, where $\mathrm{Unif}(\mathcal{Y})$ is the uniform distribution over $\mathcal{Y}$.

\smallskip{}

\begin{rem}
We may extend Definition \ref{def:game_lossless} to lossy source coding, by replacing $\mathbf{1}\{\hat{x}^{n}\neq x^{n}\}$ with $\mathbf{1}\{\sum_{i}d(x_{i},\hat{x}_{i})>nD\}$ where $d(x,\hat{x})$ is a distortion function. We may also extend it to channel simulation \cite{bennett2002entanglement,cuff2013distributed,li2024channel}, where the requirement that the channel is simulated may be written as a degradedness constraint (Definition \ref{def:degraded}). These are left for future studies. 
\end{rem}
\smallskip{}

\section{Acknowledgement}

This work was partially supported by a grant from the Research Grants Council of the Hong Kong Special Administrative Region, China {[}Project No: CUHK 14209823 (GRF){]}.

\smallskip{}

\appendix{}

\subsection{Proof of Proposition \ref{prop:double_dual}}\label{subsec:pf_double_dual}

We first prove that $\mathrm{cl}(A)=\{c\in\mathbb{R}^{\mathcal{Y}}:\,\forall c'\in\mathbb{R}^{\mathcal{Y}}.(c'<c\,\to\,c'\in A)\}$ with respect to the box topology. Consider $c\in\mathbb{R}^{\mathcal{Y}}$ such that for all $c'\in\mathbb{R}^{\mathcal{Y}}$ with $c'<c$, we have $c'\in A$. Consider an open neighborhood $\prod_{y}S(y)\ni c$ where $S:\mathcal{Y}\to\mathcal{O}$ (where $\mathcal{O}$ is the set of open subsets of $\mathbb{R}$). Let $c'(y)\in S(y)$, $c'(y)<c(y)$ for all $y$. Then $c'<c$, so $c'\in A$. Hence, $A\cap\prod_{y}S(y)\neq\emptyset$. Therefore, $c\in\mathrm{cl}(A)$. For the other direction, consider any $c\in\mathrm{cl}(A)$. Consider any $c'<c$. Consider the open neighborhood $\prod_{y}(c'(y),c(y)+1)\ni c$. Let $b\in A\cap\prod_{y}(c'(y),c(y)+1)$. Since $b\in A$ and $c'<b$, as $A$ is downward closed, $c'\in A$.

We prove that $A^{\diamondsuit}$ is closed with respect to the box topology. Assume $b\in\mathrm{cl}(A^{\diamondsuit})$.  For any $b'\in\mathbb{R}^{\mathcal{Y}}$ with $b'<b$ (pointwise), we have $b'\in A^{\diamondsuit}$.  Assume the contrary that $b\notin A^{\diamondsuit}$. There exists $a\in A$ such that $a+b>0$, $(b-a)/2<b$, so $(b-a)/2\in A^{\diamondsuit}$. However, $a+(b-a)/2=(a+b)/2>0$,  contradicting the definition of $A^{\diamondsuit}$. Hence, $b\in A^{\diamondsuit}$.

To prove $A\subseteq A^{\diamondsuit\diamondsuit}$, fix any $a\in A$. For $\tilde{a}\in A^{\diamondsuit}$, there exists $y$ such that $a(y)+\tilde{a}(y)\le0$. Hence, $a\in A^{\diamondsuit\diamondsuit}$. Since $A^{\diamondsuit\diamondsuit}$ is closed, we also have $\mathrm{cl}(A)\subseteq A^{\diamondsuit\diamondsuit}$.

We now prove $A^{\diamondsuit\diamondsuit}\subseteq\mathrm{cl}(A)$. Fix any $c\in A^{\diamondsuit\diamondsuit}$. For all $\tilde{a}\in A^{\diamondsuit}$, $c+\tilde{a}>0$ does not hold. Consider any $c'<c$. We have $-c'\notin A^{\diamondsuit}$, so there exists $a\in A$ such that $a-c'>0$, $a>c'$. Hence, $c'\in A$. Therefore, $c\in\mathrm{cl}(A)$.

$(A\cup B)^{\diamondsuit}=A^{\diamondsuit}\cap B^{\diamondsuit}$ follows from the definition. $A^{\diamondsuit}\cup B^{\diamondsuit}\subseteq(A\cap B)^{\diamondsuit}$ follows from $A^{\diamondsuit}\subseteq(A\cap B)^{\diamondsuit}$, which follows from the definition.

Finally, we prove $(A\cap B)^{\diamondsuit}\subseteq A^{\diamondsuit}\cup B^{\diamondsuit}$. Let $c\in(A\cap B)^{\diamondsuit}$.  Assume the contrary that $c\notin A^{\diamondsuit}\cup B^{\diamondsuit}$. Since $c\notin A^{\diamondsuit}$, there exists $a\in A$ such that $c+a>0$. Since $c\notin B^{\diamondsuit}$, there exists $b\in B$ such that $c+b>0$. Hence, $c(y)>-\min\{a(y),b(y)\}$ for all $y$. Since $(y\mapsto\min\{a(y),b(y)\})\in A\cap B$ and $c\in(A\cap B)^{\diamondsuit}$, there exists $y$ such that $c(y)+\min\{a(y),b(y)\}\le0$, leading to a contradiction. Hence, $c\in A^{\diamondsuit}\cup B^{\diamondsuit}$.

\subsection{Proof of Proposition \ref{prop:dual_commute}}\label{subsec:pf_dual_commute}

We prove the following slightly stronger statements.
\begin{prop}
For $A\in\mathrm{DCCone}(\mathcal{Y})$, $F:\mathcal{Y}\to\mathrm{DCCone}(\mathcal{Z})$, we have $A^{\diamondsuit}\ltimes F^{\diamondsuit}\subseteq(A\ltimes F)^{\diamondsuit}$. Equality holds if $|\mathcal{Z}|<\infty$ and $F(y)\notin\{\emptyset,\mathbb{R}^{\mathcal{Z}}\}$ for all $y$.
\end{prop}
\begin{IEEEproof}
Consider $t\in A^{\diamondsuit}\ltimes F^{\diamondsuit}$, and let $t(y,z)=b(y)+g(y,z)$ where $b\in A^{\diamondsuit}$, $g(y,\cdot)\in F^{\diamondsuit}(y)$. Consider $s\in A\ltimes F$ with $s(y,z)=a(y)+f(y,z)$, $a\in A$, $f(y,\cdot)\in F(y)$. Since $b\in A^{\diamondsuit}$, there exists $y$ such that $a(y)+b(y)\le0$. Since $g(y,\cdot)\in F^{\diamondsuit}(y)$, there exists $z$ such that $f(y,z)+g(y,z)\le0$. Hence, $s(y,z)+t(y,z)\le0$. We have $t\in(A\ltimes F)^{\diamondsuit}$.

For the other direction, consider $t\in(A\ltimes F)^{\diamondsuit}$. Let $b(y):=\sup_{h\in F(y)}\inf_{z}(h(z)+t(y,z))$, which is finite since $|\mathcal{Z}|<\infty$ and $F(y)\notin\{\emptyset,\mathbb{R}^{\mathcal{Z}}\}$. We first prove $b\in A^{\diamondsuit}$. Assume the contrary that $b\notin A^{\diamondsuit}$. Then there exists $a\in A$ such that $a+b>0$, $b>-a$. Let $f\in\mathbb{R}^{\mathcal{Y}\times\mathcal{Z}}$ with $f(y,\cdot)\in F(y)$ and $\inf_{z}(f(y,z)+t(y,z))>-a(y)$ for all $y$. Since $t\in(A\ltimes F)^{\diamondsuit}$, there exists $y,z$ such that $a(y)+f(y,z)+t(y,z)\le0$, contradicting with $\inf_{z}(f(y,z)+t(y,z))>-a(y)$. Hence, $b\in A^{\diamondsuit}$.

We then prove $(z\mapsto t(y,z)-b(y))\in F^{\diamondsuit}(y)$ for all $y$. Fix $y$ and let $v(z):=t(y,z)-b(y)$. Assume the contrary that $v\notin F^{\diamondsuit}(y)$. There exists $g\in F(y)$ such that $g+v>0$. For every $z$, we have $v(z)\le t(y,z)-\inf_{z'}(g(z')+t(y,z'))$, and hence, $g(z)+t(y,z)-\inf_{z'}(g(z')+t(y,z'))>0$. Since $\mathcal{Z}$ is finite, taking $z$ to be the minimizer of $g(z)+t(y,z)$ leads to a contradiction. Hence, $v\in F^{\diamondsuit}(y)$. Therefore, $t\in A^{\diamondsuit}\ltimes F^{\diamondsuit}$.
\end{IEEEproof}
\smallskip{}

\begin{prop}
For $A\in\mathrm{DCCone}(\mathcal{Y})$, $f:\mathcal{Y}\to\mathcal{Z}$, we have $f\#A^{\diamondsuit}\subseteq(f\#A)^{\diamondsuit}$. Equality holds if $|f^{-1}(\{z\})|<\infty$ for all $z\in\mathcal{Z}$.
\end{prop}
\begin{IEEEproof}
Consider $t\in f\#A^{\diamondsuit}$. We have $(y\mapsto t(f(y)))\in A^{\diamondsuit}$. Assume the contrary that $t\notin(f\#A)^{\diamondsuit}$. There exists $b\in f\#A$ such that $t+b>0$. We have $(y\mapsto b(f(y)))\in A$. Hence, there exists $y$ such that $b(f(y))+t(f(y))\le0$, contradicting with $t+b>0$. Therefore, $t\in(f\#A)^{\diamondsuit}$.

For the other direction, consider $t\in(f\#A)^{\diamondsuit}$. Assume the contrary that $t\notin f\#A^{\diamondsuit}$, i.e., $(y\mapsto t(f(y)))\notin A^{\diamondsuit}$. There exists $a\in A$ such that $a(y)+t(f(y))>0$ for all $y$. Let $b\in\mathbb{R}^{\mathcal{Z}}$, $b(z):=\min_{y\in f^{-1}(\{z\})}a(y)$. Since $(y\mapsto b(f(y)))\le a\in A$, we have $b\in f\#A$. We also have $b(z)+t(z)>0$ for all $z$, contradicting with $t\in(f\#A)^{\diamondsuit}$.
\end{IEEEproof}
\smallskip{}

For the remainder of Proposition \ref{prop:dual_commute}, we have $F^{\diamondsuit}\#A^{\diamondsuit}=\pi_{2}\#(A^{\diamondsuit}\ltimes F^{\diamondsuit})\subseteq\pi_{2}\#(A\ltimes F)^{\diamondsuit}\subseteq(\pi_{2}\#(A\ltimes F))^{\diamondsuit}=(F\#A)^{\diamondsuit}$. Equality holds if $|\mathcal{Y}|,|\mathcal{Z}|<\infty$ and $F(y)\notin\{\emptyset,\mathbb{R}^{\mathcal{Z}}\}$ for all $y$. Let $F:[2]\to\mathrm{DCCone}(\mathcal{Y})$, $F(1)=A$, $F(2)=B$. We have $A^{\diamondsuit}\stackrel{\lambda}{\oplus}B^{\diamondsuit}:=F^{\diamondsuit}\#((1-\lambda,\lambda)^{\circ})\subseteq(F\#((1-\lambda,\lambda)^{\circ}))^{\diamondsuit}=(A\stackrel{\lambda}{\oplus}B)^{\diamondsuit}$. Equality holds if $|\mathcal{Y}|<\infty$ and $A,B\notin\{\emptyset,\mathbb{R}^{\mathcal{Y}}\}$.

\subsection{Proof of Proposition \ref{prop:i_prop}}\label{subsec:pf_i_prop}

For $a\in\mathbb{R}^{\mathcal{Y}}$, write $a^{\circ\Delta}:=\{q\in\Delta_{\mathcal{Y}}:\,\langle q,a\rangle\le0\}$. We first prove that $\mathrm{I}(A)=0$ if and only if $\mathrm{hull}(A)\neq\mathbb{R}^{\mathcal{Y}}$. If $\mathrm{I}(A)=0$, then for any $\epsilon>0$, there exists $q\in\Delta_{\mathcal{Y}}$ such that for all $a\in A$, we have $\inf_{p\in a^{\circ\Delta}}D(p\Vert q)<\epsilon$, so there exists $p\in a^{\circ\Delta}$ (or $a\in p^{\circ}$) with $D(p\Vert q)<\epsilon$. Since $\Delta_{\mathcal{Y}}$ is compact, there exists a convergent sequence $q_{1},q_{2},\ldots\in\Delta_{\mathcal{Y}}$ and $\epsilon_{i}\to0$ such that for all $i\ge1$ and $a\in A$, there exists $p\in a^{\circ\Delta}$ with $D(p\Vert q_{i})<\epsilon_{i}$, so $\Vert p-q_{i}\Vert_{1}\le\sqrt{2\epsilon_{i}}$ by Pinsker's inequality.  Let $\lim_{i}q_{i}=q$. Assume the contrary that there exists $a\in A$, $a\notin q^{\circ}$, i.e., $\langle q,a\rangle>0$. There exists $\delta>0$ such that $\langle p,a\rangle>\langle q,a\rangle/2$ for every $p\in\Delta_{\mathcal{Y}}$ with $\Vert p-q\Vert_{1}\le\delta$, leading to a contradiction by taking $i$ large enough so $\sqrt{2\epsilon_{i}}+\Vert q_{i}-q\Vert_{1}<\delta$. Hence, we have $A\subseteq q^{\circ}$, so $\mathrm{hull}(A)\subseteq q^{\circ}\neq\mathbb{R}^{\mathcal{Y}}$.

For the other direction, assume $\mathrm{hull}(A)\neq\mathbb{R}^{\mathcal{Y}}$. If $\epsilon\mathbf{1}\in\mathrm{hull}(A)$ for some $\epsilon>0$ (where $\mathbf{1}\in\mathbb{R}^{\mathcal{Y}}$ is the all one vector), then $\mathbb{R}_{\le0}^{\mathcal{Y}}+\gamma\epsilon\mathbf{1}\subseteq\mathrm{hull}(A)$ for all $\gamma\ge0$, implying $\mathrm{hull}(A)=\mathbb{R}^{\mathcal{Y}}$. Hence, $\epsilon\mathbf{1}\notin\mathrm{hull}(A)$ for every $\epsilon>0$, so $0$ is on the boundary of $\mathrm{hull}(A)$. Let $q\in\Delta_{\mathcal{Y}}$ be the normal of a supporting hyperplane of $\mathrm{hull}(A)$ at point $0$, i.e., $\mathrm{hull}(A)\subseteq q^{\circ}$. For every $a\in A$, taking $p=q$, we have $\langle a,p\rangle\le0$ and $D(p\Vert q)=0$. So $\mathrm{I}(A)=0$. 

For the alternative characterization, we assume $A\neq\mathbb{R}^{\mathcal{Y}}$ (otherwise both sides of (\ref{eq:alt_char}) are infinite). We first prove the case where $A=\mathrm{dccone}(S)$ is generated by a finite set $S\subseteq\mathbb{R}^{\mathcal{Y}}$, where $\mathrm{dccone}(S):=\bigcup_{\gamma\ge0}\gamma S+\mathbb{R}_{\le0}^{\mathcal{Y}}$ is the smallest pricing DC cone containing $S$. In this case, it suffices to consider $p_{T}$ and $p_{Y|T}$ supported over $T\in S$. We have
\begin{align}
 & \sup_{p_{T}}\inf_{p_{Y|T}}I(T;Y)\nonumber \\
 & \stackrel{(a)}{=}\inf_{p_{Y|T}}\sup_{p_{T}}I(T;Y)\nonumber \\
 & \stackrel{(b)}{=}\inf_{p_{Y|T}}\inf_{q\in\Delta_{\mathcal{Y}}}\sup_{a\in A}D(p_{Y|T}(\cdot|a)\Vert q)\nonumber \\
 & =\inf_{q\in\Delta_{\mathcal{Y}}}\sup_{a\in A}\inf_{p\in a^{\circ\Delta}}D(p\Vert q)\;=\;\mathrm{I}(A),\label{eq:pf_alt_finite}
\end{align}
where $p_{Y|T}$ ranges over conditional distributions from $S$ to $\mathcal{Y}$ satisfying that $\langle t,p_{Y|T}(\cdot|t)\rangle\le0$ for every $t\in A$, (a) is by Sion's minimax theorem \cite{sion1958general,mceliece1983communication}, and (b) is by the information radius characterization of channel capacity \cite{kemperman1974shannon}. We now consider general $A$. Since the supremum on the right-hand side of (\ref{eq:alt_char}) is over finitely-supported $p_{T}$, combining this with (\ref{eq:pf_alt_finite}), we have 
\begin{equation}
\sup_{p_{T}}\inf_{p_{Y|T}}I(T;Y)=\sup_{S\subseteq A,\,|S|<\infty}\mathrm{I}(\mathrm{dccone}(S))\le\mathrm{I}(A).\label{eq:pf_alt_sup}
\end{equation}
It is left to construct $S$ that ``covers'' $A$ so that $\mathrm{I}(\mathrm{dccone}(S))\approx\mathrm{I}(A)$. Fix any $\epsilon>0$. Let $S\subseteq A$ be a finite set such that for every $a\in A$, there exists $b\in S$ with $d(a^{\circ\Delta},b^{\circ\Delta})\le\epsilon$, where $d(a^{\circ\Delta},b^{\circ\Delta})$ is the Hausdorff metric between $a^{\circ\Delta}$ and $b^{\circ\Delta}$. This $S$ can be constructed using the standard metric covering construction, i.e., starting with $S=\emptyset$, and iteratively adding any $a\in A$ with $\min_{b\in S}d(a^{\circ\Delta},b^{\circ\Delta})>\epsilon$ into $S$. To show that this process terminates, since $\Delta_{\mathcal{Y}}$ is compact, the metric space of nonempty compact subsets of $\Delta_{\mathcal{Y}}$ with the Hausdorff metric is compact as well, and hence can be covered by finitely many metric balls of radius $\epsilon/2$. If the process does not terminate, then there must exist $b_{1},b_{2}\in S$ where $b_{1}^{\circ\Delta}$ and $b_{2}^{\circ\Delta}$ are in the same metric ball, leading to a contradiction. Hence, we can construct a finite $S$. Let $u\in\mathbb{R}^{\mathcal{Y}}$, $u(y)=1/|\mathcal{Y}|$. Fix any $0<\delta<1$. We have
\begin{align*}
\mathrm{I}(A) & =\inf_{q\in\Delta_{\mathcal{Y}}}\sup_{a\in A}\inf_{p\in a^{\circ\Delta}}D(p\Vert q)\\
 & \le\inf_{q\in\Delta_{\mathcal{Y}}}\sup_{a\in A}\inf_{p\in a^{\circ\Delta}}D(p\Vert(1-\delta)q+\delta u)\\
 & \stackrel{(a)}{\le}\inf_{q\in\Delta_{\mathcal{Y}}}\sup_{a\in S}\inf_{p\in a^{\circ\Delta}}D(p\Vert(1-\delta)q+\delta u)+c_{\delta,\epsilon}\\
 & \le\inf_{q\in\Delta_{\mathcal{Y}}}\sup_{a\in S}\inf_{p\in a^{\circ\Delta}}((1-\delta)D(p\Vert q)+\delta D(p\Vert u))+c_{\delta,\epsilon}\\
 & \le\inf_{q\in\Delta_{\mathcal{Y}}}\sup_{a\in S}\inf_{p\in a^{\circ\Delta}}D(p\Vert q)+\delta\log|\mathcal{Y}|+c_{\delta,\epsilon}\\
 & =\mathrm{I}(\mathrm{dccone}(S))+\delta\log|\mathcal{Y}|+c_{\delta,\epsilon}\\
 & \stackrel{(b)}{=}\sup_{p_{T}}\inf_{p_{Y|T}}I(T;Y)+\delta\log|\mathcal{Y}|+c_{\delta,\epsilon},
\end{align*}
where (a) is because $(p,q)\mapsto D(p\Vert(1-\delta)q+\delta u)$ is uniformly continuous by the Heine-Cantor theorem, i.e., there exists $c_{\delta,\epsilon}>0$ with $\lim_{\epsilon\to0}c_{\delta,\epsilon}=0$ such that $|D(p\Vert(1-\delta)q+\delta u)-D(p'\Vert(1-\delta)q'+\delta u)|\le c_{\delta,\epsilon}$ as long as $\Vert p-p'\Vert+\Vert q-q'\Vert\le\epsilon$, and (b) is due to (\ref{eq:pf_alt_sup}). The proof is completed by taking $\epsilon\to0$, $\delta\to0$.

To show that $\mathrm{I}(A)$ generalizes Shannon capacity, consider $A=\bigcup_{x}(p_{Y|X}(\cdot|x))^{\circ}$.  We have
\begin{align*}
\mathrm{I}(A) & =\inf_{q\in\Delta_{\mathcal{Y}}}\sup_{a\in A}\inf_{p\in a^{\circ\Delta}}D(p\Vert q)\\
 & =\inf_{q\in\Delta_{\mathcal{Y}}}\max_{x}\sup_{a\in p_{Y|X}(\cdot|x)^{\circ}}\inf_{p\in a^{\circ\Delta}}D(p\Vert q)\\
 & \stackrel{(a)}{=}\inf_{q\in\Delta_{\mathcal{Y}}}\max_{x}D(p_{Y|X}(\cdot|x)\Vert q)\\
 & =\max_{p_{X}}I(X;Y),
\end{align*}
where (a) is because $D(p\Vert q)$ is convex in $p$, so we can consider a subgradient $v\in\mathbb{R}^{\mathcal{Y}}$ such that $D(p\Vert q)\ge D(p_{Y|X}(\cdot|x)\Vert q)+\langle p-p_{Y|X}(\cdot|x),v\rangle$ for all $p\in\Delta_{\mathcal{Y}}$. Taking $a(y)=\langle p_{Y|X}(\cdot|x),v\rangle-v(y)$, we have $a\in p_{Y|X}(\cdot|x)^{\circ}$ and $D(p\Vert q)\ge D(p_{Y|X}(\cdot|x)\Vert q)-\langle p,a\rangle\ge D(p_{Y|X}(\cdot|x)\Vert q)$ for $p\in a^{\circ\Delta}$.

For the additivity property, we first prove $\mathrm{I}(A\ltimes B)\le\mathrm{I}(A)+\mathrm{I}(B)$. Fix any $q_{Y}\in\Delta_{\mathcal{Y}}$, $q_{Z}\in\Delta_{\mathcal{Z}}$. For any $s\in A\ltimes B$, let $a\in A$ and $f\in\mathbb{R}^{\mathcal{Y}\times\mathcal{Z}}$ with $f(y,\cdot)\in B$ for all $y$, such that $s(y,z)=a(y)+f(y,z)$. Note that for $p\in\Delta_{\mathcal{Y}\times\mathcal{Z}}$, the condition ``$p_{Y}\in a^{\circ\Delta}$ and $p_{Z|Y}(\cdot|y)\in f(y,\cdot)^{\circ\Delta}$ for all $y$'' (where $p_{Y}(y)=\sum_{z}p(y,z)$ and $p_{Z|Y}(z|y)=p(y,z)/p_{Y}(y)$)\footnote{If $p_{Y}(y)=0$, take an arbitrary distribution for $p_{Z|Y}(\cdot|y)$.} implies the condition ``$p\in s^{\circ\Delta}$''. Hence,
\begin{align*}
 & \inf_{p\in s^{\circ\Delta}}D(p\Vert q_{Y}\times q_{Z})\\
 & \le\inf_{p_{Y}\in a^{\circ\Delta},\,p_{Z|Y}(\cdot|y)\in f(y,\cdot)^{\circ\Delta}}D(p_{Y}p_{Z|Y}\Vert q_{Y}\times q_{Z})\\
 & =\inf_{p_{Y}\in a^{\circ\Delta},\,p_{Z|Y}(\cdot|y)\in f(y,\cdot)^{\circ\Delta}}\Big(D(p_{Y}\Vert q_{Y})+\sum_{y}p_{Y}(y)D(p_{Z|Y}(\cdot|y)\Vert q_{Z})\Big)\\
 & =\inf_{p_{Y}\in a^{\circ\Delta}}\Big(D(p_{Y}\Vert q_{Y})+\sum_{y}p_{Y}(y)\inf_{p_{Z|Y}(\cdot|y)\in f(y,\cdot)^{\circ\Delta}}D(p_{Z|Y}(\cdot|y)\Vert q_{Z})\Big)\\
 & \le\inf_{p_{Y}\in a^{\circ\Delta}}\Big(D(p_{Y}\Vert q_{Y})+\sum_{y}p_{Y}(y)\sup_{b\in B}\inf_{p_{Z|Y}(\cdot|y)\in b^{\circ\Delta}}D(p_{Z|Y}(\cdot|y)\Vert q_{Z})\Big)\\
 & =\inf_{p_{Y}\in a^{\circ\Delta}}D(p_{Y}\Vert q_{Y})+\sup_{b\in B}\inf_{p_{Z}\in b^{\circ\Delta}}D(p_{Z}\Vert q_{Z}).
\end{align*}
Therefore,
\begin{align*}
 & \sup_{s\in A\ltimes B}\inf_{p\in s^{\circ\Delta}}D(p\Vert q_{Y}\times q_{Z})\\
 & \le\sup_{a\in A}\inf_{p_{Y}\in a^{\circ\Delta}}D(p_{Y}\Vert q_{Y})+\sup_{b\in B}\inf_{p_{Z}\in b^{\circ\Delta}}D(p_{Z}\Vert q_{Z}),
\end{align*}
which gives $\mathrm{I}(A\ltimes B)\le\mathrm{I}(A)+\mathrm{I}(B)$.

We then prove $\mathrm{I}(A\otimes B)\ge\mathrm{I}(A)+\mathrm{I}(B)$ using (\ref{eq:alt_char}). Let $T_{A}\in A$ and $T_{B}\in B$ be independent random variables with finite supports. Let
\begin{align*}
f(\alpha,\beta) & :=\inf_{P_{Y,Z|T_{A},T_{B}}:\,\mathbb{E}[T_{A}(Y)]\le\alpha,\,\mathbb{E}[T_{B}(Z)]\le\beta}I(T_{A},T_{B};Y,Z)\\
 & \ge\inf_{P_{Y,Z|T_{A},T_{B}}:\,\mathbb{E}[T_{A}(Y)]\le\alpha,\,\mathbb{E}[T_{B}(Z)]\le\beta}(I(T_{A};Y)+I(T_{B};Z))\\
 & \ge\inf_{P_{Y|T_{A}}:\,\mathbb{E}[T_{A}(Y)]\le\alpha}I(T_{A};Y)+\inf_{P_{Z|T_{B}}:\,\mathbb{E}[T_{B}(Z)]\le\beta}I(T_{B};Z).
\end{align*}
Note that $f$ is convex in $(\alpha,\beta)\in\mathbb{R}^{2}$. Let $-v\in\mathbb{R}^{2}$ be a subgradient of $f$ at $(0,0)$, i.e., $f(\alpha,\beta)\ge f(0,0)-v_{1}\alpha-v_{2}\beta$ for all $(\alpha,\beta)\in\mathbb{R}^{2}$. By (\ref{eq:alt_char}), we have
\begin{align*}
 & \mathrm{I}(A\otimes B)\\
 & \ge\inf_{P_{Y,Z|T_{A},T_{B}}:\,\mathbb{E}[v_{1}T_{A}(Y)+v_{2}T_{B}(Z)]\le0}I(T_{A},T_{B};Y,Z)\\
 & =\inf_{\alpha,\beta:\,v_{1}\alpha+v_{2}\beta\le0}f(\alpha,\beta)\\
 & =f(0,0)\\
 & \ge\inf_{P_{Y|T_{A}}:\,\mathbb{E}[T_{A}(Y)]\le0}I(T_{A};Y)+\inf_{P_{Z|T_{B}}:\,\mathbb{E}[T_{B}(Z)]\le0}I(T_{B};Z).
\end{align*}
Hence, $\mathrm{I}(A\otimes B)\ge\mathrm{I}(A)+\mathrm{I}(B)$.

Before proving the monotonicity property, we first prove $\mathrm{I}(f\#A)\le\mathrm{I}(A)$ for $f:\mathcal{Y}\to\mathcal{Z}$. Fix any $q\in\Delta_{\mathcal{Y}}$. Let $q'=f\#q\in\Delta_{\mathcal{Z}}$. Fix any $b\in f\#A$, i.e., $(y\mapsto b(f(y)))\in A$. Let $a=(y\mapsto b(f(y)))$. For any $p\in a^{\circ\Delta}$, taking $p'=f\#p$, we have $\langle p',b\rangle=\sum_{y}p(y)b(f(y))=\langle p,a\rangle\le0$, and $D(p'\Vert q')\le D(p\Vert q)$ by data processing. The result follows. 

We now prove the monotonicity property. Assume $B\preceq A$, i.e., there exists  $F:\mathcal{Y}\to\mathrm{DCCone}(\mathcal{Z})$ with $F(y)+F(y)\neq\mathbb{R}^{\mathcal{Z}}$ such that $B\subseteq F\#A=\pi_{2}\#(A\ltimes F)$. We first show that $\mathrm{I}(A\ltimes F)\le\mathrm{I}(A)$. Fix any $q_{Y}\in\Delta_{\mathcal{Y}}$. Let $q_{Z|Y}$ be a conditional distribution from $\mathcal{Y}$ to $\mathcal{Z}$ such that $\langle q_{Z|Y}(\cdot|y),b\rangle\le0$ for all $y\in\mathcal{Y}$, $b\in F(y)$ (we can take $q_{Z|Y}(\cdot|y)$ to be arbitrary if $F(y)=\emptyset$, or otherwise be the normal of the supporting hyperplane of $F(y)+F(y)$ at $0$, which must have nonnegative entries since $\mathbb{R}_{\le0}^{\mathcal{Z}}\subseteq F(y)$). For any $s\in A\ltimes F$, let $a\in A$ and $f\in\mathbb{R}^{\mathcal{Y}\times\mathcal{Z}}$ with $f(y,\cdot)\in F(y)$ for all $y$, such that $s(y,z)=a(y)+f(y,z)$. Note that for $p\in\Delta_{\mathcal{Y}\times\mathcal{Z}}$, the condition ``$p_{Y}\in a^{\circ\Delta}$ and $p_{Z|Y}(\cdot|y)\in f(y,\cdot)^{\circ\Delta}$ for all $y$'' implies the condition ``$p\in s^{\circ\Delta}$''. Hence,
\begin{align*}
 & \inf_{p\in s^{\circ\Delta}}D(p\Vert q_{Y}q_{Z|Y})\\
 & \le\inf_{p_{Y}\in a^{\circ\Delta},\,p_{Z|Y}(\cdot|y)\in f(y,\cdot)^{\circ\Delta}}D(p_{Y}p_{Z|Y}\Vert q_{Y}q_{Z|Y})\\
 & \stackrel{(a)}{=}\inf_{p_{Y}\in a^{\circ\Delta}}D(p_{Y}\Vert q_{Y})\\
 & \le\sup_{a\in A}\inf_{p_{Y}\in a^{\circ\Delta}}D(p_{Y}\Vert q_{Y}),
\end{align*}
where (a) is because we can take $p_{Z|Y}=q_{Z|Y}$ since $\langle q_{Z|Y}(\cdot|y),f(y,\cdot)\rangle\le0$. Hence, $\mathrm{I}(A\ltimes F)\le\mathrm{I}(A)$. We have $\mathrm{I}(B)\le\mathrm{I}(\pi_{2}\#(A\ltimes F))\le\mathrm{I}(A\ltimes F)\le\mathrm{I}(A)$.

\smallskip{}

\bibliographystyle{IEEEtran}
\bibliography{ref}

\end{document}